\RequirePackage{fix-cm}
\documentclass[sglanonrev]{informs4}

\RequirePackage{tgtermes}
\RequirePackage{newtxtext}
\RequirePackage{newtxmath}
\RequirePackage{bm}
\RequirePackage{endnotes}
\usepackage{microtype}
\usepackage{xcolor}
\renewcommand{\theARTICLETOP}{}

\OneAndAHalfSpacedXII

\usepackage{graphicx}
\usepackage{amsmath}
\usepackage[linesnumbered,ruled,vlined]{algorithm2e}
\usepackage[nolist]{acronym}
\usepackage{comment}
\usepackage{booktabs}
\usepackage{silence}
\usepackage[compatibility=false]{caption}
\usepackage{subcaption}
\usepackage{tikz}
\usepackage[section]{placeins}
\usepackage{natbib}
\bibpunct[, ]{(}{)}{,}{a}{}{,}%
\def\bibfont{\fontsize{11}{11}\selectfont}%
\usepackage[colorlinks=true,citecolor = blue, linkcolor = blue,urlcolor  = blue]{hyperref}

\EquationsNumberedThrough

\TheoremsNumberedThrough
\ECRepeatTheorems

\MANUSCRIPTNO{}

\makeatletter
\@ifundefined{assumption}{\newtheorem{assumption}{Assumption}}{}
\@ifundefined{proposition}{}{}
\@ifundefined{corollary}{\newtheorem{corollary}{Corollary}}{}
\@ifundefined{definition}{}{}
\@ifundefined{remark}{\newtheorem{remark}{Remark}}{}
\makeatother

\providecommand{\bA}{\boldsymbol{A}}
\providecommand{\ba}{\boldsymbol{a}}
\providecommand{\bB}{\boldsymbol{B}}
\providecommand{\bD}{\boldsymbol{D}}
\providecommand{\bE}{\boldsymbol{E}}
\providecommand{\bI}{\boldsymbol{I}}
\providecommand{\bJ}{\boldsymbol{J}}
\providecommand{\bP}{\boldsymbol{P}}
\providecommand{\bPhi}{\boldsymbol{\Phi}}
\providecommand{\bX}{\boldsymbol{X}}
\providecommand{\bb}{\boldsymbol{b}}
\providecommand{\bc}{\boldsymbol{c}}
\providecommand{\bg}{\boldsymbol{g}}
\providecommand{\blambda}{\boldsymbol{\lambda}}
\providecommand{\bmu}{\boldsymbol{\mu}}
\providecommand{\bnu}{\boldsymbol{\nu}}
\providecommand{\bpi}{\boldsymbol{\pi}}
\providecommand{\br}{\boldsymbol{r}}
\providecommand{\bv}{\boldsymbol{v}}
\providecommand{\bw}{\boldsymbol{w}}
\providecommand{\by}{\boldsymbol{y}}
\providecommand{\bx}{\boldsymbol{x}}
\providecommand{\bzero}{\boldsymbol{0}}
\providecommand{\bDelta}{\boldsymbol{\Delta}}
\providecommand{\R}{\mathbb{R}}
\providecommand{\E}{\mathbb{E}}
\providecommand{\cS}{\mathcal{S}}
\providecommand{\cA}{\mathcal{A}}
\providecommand{\cB}{\mathcal{B}}
\providecommand{\cN}{\mathcal{N}}
\providecommand{\qedhere}{}

\def\graphscaler{1.0}
\def\edits#1{{\color{black}#1}}

\begin{document}
\hfuzz=10000pt
\vfuzz=10000pt
\hbadness=10000
\vbadness=10000

\acrodef{MDP}{Markov Decision Process}
\acrodef{MDPs}{Markov Decision Processes}
\acrodef{ADP}{Approximate Dynamic Programming}
\acrodef{DP}{Dynamic Programming}
\acrodef{LP}{Linear Programming}
\acrodef{ALP}{Approximate Linear Programming}
\acrodef{DALP}{Dual Approximate Linear Programming}
\acrodef{VFA}{Value Function Approximation}
\acrodef{PIA}{Primal Iterative Algorithm}
\acrodef{DIA}{Dual Iterative Algorithm}
\acrodef{CI}{Confidence Interval}

\RUNAUTHOR{Anonymous}

\RUNTITLE{Dual-Based Weight Selection for \ac{ALP}}

\TITLE{Dual-Based Weight Selection for Approximate Linear Programming}

\ARTICLEAUTHORS{%
\AUTHOR{Su Li}
\AFF{Department of Computational Applied Mathematics \& Operations Research, Rice University, \EMAIL{sl355@rice.edu}}
\AUTHOR{Andre A. Cire}
\AFF{Operations Management and Statistics, University of Toronto, \EMAIL{andre.cire@Rotman.utoronto.ca}}
\AUTHOR{Adam Diamant}
\AFF{Schulich School of Business, York University, \EMAIL{adiamant@schulich.yorku.ca}}
\AUTHOR{Vahid Sarhangian}
\AFF{Department of Mechanical \& Industrial Engineering, University of Toronto, \EMAIL{v.sarhangian@utoronto.ca}}
}

\acresetall

\ABSTRACT{%
\edits{\ac{ALP} is widely used for large-scale \ac{MDPs}, but its performance can be sensitive to the choice of state-relevance weights, which are typically selected heuristically. Performance bounds suggest aligning these weights with the discounted occupancy measure of the induced policy, and existing primal approaches address this through repeated greedy-policy construction. Nonetheless, they lack convergence guarantees and are computationally expensive. We propose a dual-based method that uses projected occupancy information from the \ac{ALP} dual solution to construct a smooth stochastic policy and update the state-relevance weights, which avoids separate greedy-action calculations. We establish conditions under which the weights match the discounted occupancy of the induced policy and prove uniqueness and global convergence under appropriate smoothing. We also derive an a posteriori policy-loss bound that separates error from the weighted Bellman residual, occupancy mismatch, and stochastic-versus-greedy disagreement. Experiments on classical queueing and multi-priority scheduling problems show that the proposed approach reduces sensitivity to fixed weights and achieves comparable or better policy quality than primal updates at lower computational cost. Finally, we show that adaptive weighting is most valuable when the basis functions are sufficiently expressive for occupancy information to influence the resulting policy.
}}



\maketitle
\acresetall

\IfFileExists{\jobname.ent}{%
\begingroup
\parindent 0pt
\parskip 0.0ex
\def\enotesize{\normalsize}
\theendnotes
\endgroup
}{}

\section{Introduction}
\label{sec:introduction}

Markov decision processes (MDPs) provide a general framework for sequential decision making, with applications in service and healthcare operations  \citep{maxwell2010approximate, dong2025multiclass}, appointment scheduling \citep{patrick2008dynamic,saure2012dynamic,diamant2021dynamic}, inventory control \citep{abouee2026platelet,lin2020revisiting}, and vehicle routing \citep{secomandi2009reoptimization}, among others. However, computing optimal policies for realistically-sized MDPs is often intractable because exact methods must account for large state and action spaces \citep{puterman2014markov}. A widely applied and extensively studied approach \edits{for overcoming this issue is} \ac{ALP} \citep{schweitzer1985generalized,de2003linear}, which approximates the optimal value function as a linear combination of prespecified basis functions. The coefficients of these basis functions are determined by solving a linear program constructed from the Bellman inequalities, and the resulting value-function approximation can then be used to construct approximate policies.

This linear \edits{value function} approximation, however, introduces an additional modeling choice. Because the basis-function representation couples the approximated values across states through a common set of coefficients, these values cannot be optimized independently for every state, as in the exact MDP formulation. The decision maker must therefore specify the relative importance of the states in the underlying \ac{ALP} model, which is done by assigning \textit{state-relevance weights} in the objective of the associated linear program. In much of the \ac{ALP} literature, the state-relevance weights are specified a priori or numerically; common choices include, e.g., uniformly-chosen weights~\citep{le2004choose} or the initial state distribution \citep{saure2012dynamic}.

Although straightforward to implement, these heuristic choices provide limited guidance on how the resulting policies \edits{induced by the \ac{ALP} solution} will perform in practice, and their specification often depends on undocumented assumptions or implementation details. These concerns have motivated a stream of research mitigating the sensitivity of \ac{ALP} solutions to state-relevance weights. For instance, \citet{taylor2014analysis} analyze how state-relevance weights and constraint-sampling distributions affect the approximation accuracy of regularized \ac{ALP} formulations, whereas \citet{beuchat2016alleviating} combine approximations obtained under multiple weight specifications to reduce this sensitivity. \edits{More recently, \citet{pakiman2025self} develop a self-guided \ac{ALP} that progressively introduces random basis functions and uses guiding constraints to implicitly adapt the state-relevance weights, thereby reducing the importance of the initial choice of state-relevance distribution.}

Given a fixed set of basis functions, the key challenge is the potential misalignment between the state-relevance weights and the state-occupancy distribution of the policy they induce, which is central to existing \ac{ALP} performance bounds \citep{de2003linear, de2008choosing}. This alignment is difficult to model and optimize because of the complex interdependence between the state-relevance weights and the induced policy. For example, repeated policy simulation may be required to evaluate the effects of a candidate weight vector. The literature directly addressing this problem, however, is limited. In an unpublished manuscript, \citet{le2004choose} formulate a fixed-point algorithm that iteratively solves the \ac{ALP}, computes the state-occupancy distribution of the resulting policy, and uses this distribution to update the state-relevance weight vector. \edits{Similarly, \citet{de2008choosing} establish the existence of state-relevance weight vectors that satisfy a family of fixed-point equations and demonstrate a stochastic approximation procedure for computing them numerically. However, neither work establishes convergence to a desired fixed point, and both are computationally challenging to derive for problems with large action spaces.}

\smallskip
\noindent 
\textit{Contributions.} We address these limitations by developing a dual-based procedure for selecting state-relevance weights with convergence and ex-post performance guarantees. \edits{Our objective is to provide a standardized approach that is simple to implement, reproducible, and applicable across problems. The method leverages the dual variables of the \ac{ALP}, which satisfy projected flow-balance constraints and therefore retain partial information about state-action occupancy.}

\edits{At each iteration, our procedure first solves the dual of the ALP. It then constructs a smoothed stochastic policy by combining the selected dual solution with a prespecified target policy through a softmax operator, updating the state-relevance weights using the resulting discounted occupancy measure. We show that a fixed point of this update coincides with the occupancy measure of the induced stochastic dual policy.} The proof leverages the stability of the selected dual solution with respect to the state-relevance weights and Lipschitz continuity introduced by policy smoothing. These properties make the resulting occupancy mapping continuous and ensure the existence of a fixed point. Under smoothing conditions, it is also contractive, and therefore unique with global geometric convergence. By constructing the policy directly from the dual solution, the procedure also avoids the separate greedy-policy calculation required by existing primal-based updates.

\edits{We evaluate our approach on queueing and multi-priority scheduling problems, two settings in which \ac{ALP} has been widely applied. The results show that our algorithm matches or improves upon existing primal-based methods at substantially lower computational cost. More importantly, they demonstrate that adaptive selection of state-relevance weights can materially improve policy quality, with large performance gaps relative to heuristically chosen weighting schemes. The magnitude of these gains, however, depends on the expressivity of the basis functions. That is, with limited number of basis functions, performance is dominated by approximation error of the ALP, whereas richer architectures allow adaptive weights to more effectively influence the resulting policy.
}

\edits{The paper is organized as follows. Section~\ref{sec:problem_description} formulates the weight-selection problem, and Section~\ref{sec:dual_iterative_approach} develops the proposed dual-based approach. Section~\ref{sec:convergence_performance_analysis} establishes its fixed-point, convergence, and performance properties. Section~\ref{sec:numerical_experiments} presents the numerical experiments, and Section~\ref{sec:conclusion} concludes.}

\section{Problem Description}
\label{sec:problem_description}

We consider a discrete-time, infinite-horizon, discounted MDP with finite state space $\mathcal S$. For each state $x\in\mathcal S$, let $\mathcal A(x)$ denote the finite set of admissible actions. When the system is in state $x \in \mathcal S$, the decision maker selects an action $a \in \mathcal A(x)$, incurs the immediate cost $g(x,a) \in \mathbb{R}$, and transitions to state $y\in\mathcal S$ with probability $P(y\mid x,a)$. Future costs are discounted by a factor $\alpha\in(0,1)$. 

Decisions are made according to a stationary randomized policy $\pi$. Specifically, let $X_t$ and $A_t$ denote the state and action at time $t$, respectively. Conditional on $X_t=x$ for some $x \in \mathcal S$, action $A_t$ is drawn from $\pi(\cdot\mid x)$, and the next state $X_{t+1}$ is drawn from $P(\cdot\mid x,A_t)$. The optimal value $J^*(x)$ is the minimum expected discounted cost when the initial state at \edits{time $t=0$ is $x$:}
\begin{align}
  \label{eq:optimal_value}
  J^*(x)
  =
  \min_{\pi}
  \,
  \mathbb E_\pi\left[
      \sum_{t=0}^{\infty}\alpha^t g(X_t,A_t)
      \,\middle|\,
      X_0=x
  \right],
  \;\; \forall x\in\mathcal S,
\end{align}
where the minimum is taken over all stationary randomized policies. Because \edits{we assume} the state and action spaces are finite and $\alpha\in(0,1)$, an optimal stationary deterministic policy exists \citep{puterman2014markov}.
Equivalently, $J^*$ is the unique solution to the Bellman optimality equations
\begin{equation}
\label{eq:bellman_optimality}
  J^*(x)
  =
  \min_{a\in\mathcal A(x)}
  \left\{
  g(x,a)
  +
  \alpha\sum_{y\in\mathcal S}P(y\mid x,a)J^*(y)
\right\},
\;\; \forall x\in\mathcal S.
\end{equation}
Although these equations characterize $J^*$ exactly, solving them becomes computationally prohibitive when the state and action spaces are large. \edits{This motivates the use of \ac{ALP}, which addresses this difficulty by approximating $J^*$ within a lower-dimensional function class and replacing the Bellman equations with linear inequalities.} Specifically, let $\phi_1,\ldots,\phi_K$ be a collection of basis functions, where $\phi_i : \mathcal S \rightarrow \mathbb{R}$. For a given coefficient vector $\br=(r_1,\ldots,r_K) \in \mathbb{R}^K$, the value at state $x$ is approximated by
the sum $
\sum_{k=1}^K r_k\phi_k(x)
\approx 
  J^*(x).
$
Thus, rather than estimating a separate value for every state, the approximation is determined by only the $K$ coefficients of $\br$.

\edits{The $|\mathcal S|$ Bellman equations in \eqref{eq:bellman_optimality} assign a distinct value to each state. In contrast, the \ac{ALP} approximates these $|\mathcal S|$ values using only $K$ coefficients, where typically $K\ll|\mathcal S|$.} This dimensionality reduction implies that the approximate values cannot be adjusted independently across states, and reducing the error in one region of the state space need not reduce it elsewhere. To choose $\br$ among these feasible approximations, the \ac{ALP} requires a state-relevance weight vector $\bc\in\mathbb R^{|\mathcal S|}$, where each component $c(x)\geq 0$
determines the emphasis placed on state $x$, with the weights typically
normalized so that $\sum_{x\in\mathcal S}c(x)=1$. That is, given $\bc$, the
\ac{ALP} selects $\br$ by solving the linear program
\begin{equation}
\label{eq:alp_primal_expanded}
\tag{pALP}
\begin{aligned}
\max_{\br\in\mathbb R^K}\quad
&
\sum_{x\in\mathcal S}c(x)
\sum_{k=1}^K r_k\phi_k(x) \\
\text{s.t.}\quad
&
\sum_{k=1}^K r_k\phi_k(x)
\leq
g(x,a)
+
\alpha\sum_{y\in\mathcal S}P(y\mid x,a)
\sum_{k=1}^K r_k\phi_k(y),
&& \forall x\in\mathcal S, \forall a \in \mathcal A(x).
\end{aligned}
\end{equation}
Each constraint requires the approximate value at state $x$ to be no greater than the one-period cost plus the discounted approximate value of the next state, for every admissible action. The objective then makes this lower approximation as large as possible. 
\edits{Given an optimal solution $\br^*$ to \eqref{eq:alp_primal_expanded}, the corresponding value function approximation (VFA) is $\widehat J_{\bc}(x) = \sum_{k=1}^K r^*_k(\bc)\phi_k(x)$. Thus, a standard approach is to construct a greedy policy that, in each state, selects the action minimizing the one-period cost plus the discounted approximate value of the next state with respect to $\widehat J_{\bc}(x)$:}
\begin{equation}
\label{eq:greedy_policy_from_alp}
u_{\bc}(x)
\in
\arg\min_{a\in\mathcal A(x)}
\left\{
g(x,a)
+
\alpha\sum_{y\in\mathcal S}
P(y\mid x,a)\widehat J_{\bc}(y)
\right\},
\;\; \forall x\in\mathcal S.
\end{equation}

\edits{Ideally, $\bc$ would be selected to minimize the expected discounted cost of the greedy policy in \eqref{eq:greedy_policy_from_alp}. However, because policy performance depends on $\bc$ in a complex way, direct optimization is generally impractical. As a result, state-relevance weights are often chosen using simple heuristics, such as uniform weights, the initial-state distribution, or problem-specific rules. Rather than relying on such heuristics,} \citet{de2003linear} suggest choosing the weights to reflect the discounted occupancy of the induced policy. Specifically, let $\bnu$ denote the initial-state distribution over $\mathcal S$. For a stationary randomized policy $\pi$, define its induced transition probabilities by
$
P_\pi(x,y)
=
\sum_{a\in\mathcal A(x)}
\pi(a\mid x)P(y\mid x,a)
$
for all $x,y\in\mathcal S$. The normalized discounted occupancy of $\pi$ is
\begin{equation}
\label{eq:discounted_occupancy}
\mu_{\pi,\bnu}(x)
:=
(1-\alpha)
\sum_{t=0}^{\infty}\alpha^t
\sum_{z\in\mathcal S}\nu(z)P_\pi^t(z,x),
\qquad x\in\mathcal S,
\end{equation}
where $P_\pi^t(z,x)$ is the probability of reaching state $x$ after $t$ periods from state $z$ under policy $\pi$, with the initial state $z$ weighted according to $\bnu$. Thus, $\mu_{\pi,\bnu}(x)$ measures the discounted frequency with which state $x$ is visited under policy $\pi$. \edits{The rationale of \citet{de2003linear} is that the \ac{ALP} places greater emphasis on approximation accuracy in states with larger state-relevance weights. However, policy performance depends most strongly on states that the induced policy visits frequently, as captured by its discounted occupancy. Choosing $\bc$ to approximate this occupancy therefore focuses the approximation on the states that matter most for policy performance.}

Based on this rationale, \citet{le2004choose} propose, in an unpublished manuscript, an iterative procedure that repeatedly solves \eqref{eq:alp_primal_expanded}, constructs a greedy policy \eqref{eq:greedy_policy_from_alp}, computes the resulting discounted occupancy measure, and uses it to update the state-relevance weights. We refer to this method as \ac{PIA}; algorithmic details and performance bounds are provided in Appendix~\ref{app:pia_algorithm}. However, \ac{PIA} is not guaranteed to converge, and its implementation requires simulating \edits{the greedy policy rollouts \eqref{eq:greedy_policy_from_alp} for every candidate $\bc$ to evaluate the discounted occupancies. This can be computationally costly for large state and action spaces.} 


\section{Dual-Based Selection of State-Relevance Weights}
\label{sec:dual_iterative_approach} 

The approach of \citet{le2004choose} motivates a general successive-approximation framework for selecting state-relevance weights. Starting from an initial $\bc^{(0)}$, each iteration generates a policy based on the current weights, computes the normalized discounted occupancy of this policy, and uses a weight-update rule to obtain the next vector. The policy generator determines how information obtained under $\bc^{(m)}$ translates into a policy, whereas the weight-update rule determines how the resulting occupancy is incorporated into $\bc^{(m+1)}$. Algorithm~\ref{alg:general_weight_framework} summarizes the framework; the returned weights are then used to define the VFA for a greedy policy.

\begin{algorithm}[t!]
\small
\caption{General Occupancy-Based Weight-Selection Framework}
\label{alg:general_weight_framework}
\SetAlgoLined
\DontPrintSemicolon
\KwIn{Initial weights $\bc^{(0)}\in\Delta(\mathcal S)$; policy generator $\mathcal G$; weight-update rule $\mathcal U$; initial-state distribution $\bnu$; tolerance $\varepsilon>0$; maximum number of updates $M$}
\KwOut{State-relevance weight vector $\bc$}
\For{$m\leftarrow 0,1,\ldots,M$}{
  Generate a policy $\pi^{(m)}\leftarrow\mathcal G(\bc^{(m)})$\;
  Compute its normalized discounted occupancy
  $\bmu^{(m)}\leftarrow\bmu_{\pi^{(m)},\bnu}$\;
  \If{$\|\bmu^{(m)}-\bc^{(m)}\|_\infty\leq\textnormal{given precision}$}{
    \Return{$\bc^{(m)}$}\;
  }
  Update the weights
  $\bc^{(m+1)}\leftarrow
  \mathcal U(\bc^{(m)},\bmu^{(m)})$\;
}
\end{algorithm}

The framework therefore leaves two components to be specified: the
policy generator $\mathcal G$ and the weight-update rule $\mathcal U$.
\edits{It also assumes the existence of a \textit{self-consistent} weight vector $\bc$, i.e., a vector that is equal to the normalized discounted occupancy of the policy generated from $\bc$. We next define both components using the dual \ac{ALP}. Specifically, the dual solution generates a smooth randomized policy, and its normalized discounted occupancy determines the updated state-relevance weights. We refer to the resulting procedure as the \ac{DIA} to distinguish it from \ac{PIA}. Smoothing provides conditions under which the updates converge, while extracting the policy directly from the dual solution avoids repeated greedy-action calculations.}

\smallskip
\noindent 
\textit{Policy Generation.} We assume that the constant function belongs to the span of the basis functions; without loss of generality, take $\phi_1(x)=1$ for all
$x\in\mathcal S$. With each state-action constraint in \eqref{eq:alp_primal_expanded} we associate a nonnegative dual variable $\lambda(x,a)$, and write the dual linear program of \ac{ALP} as
\begin{equation}
\label{eq:alp_dual_expanded}
\tag{dALP}
\begin{aligned}
  \min_{\blambda\geq\bzero}\quad
  &
  \sum_{x\in\mathcal S}
  \sum_{a\in\mathcal A(x)}
  g(x,a)\lambda(x,a) \\
  \text{s.t.}\quad
  &
  \sum_{x\in\mathcal S}
  \sum_{a\in\mathcal A(x)}
  \lambda(x,a)
  \left[
  \phi_k(x)
  -
  \alpha\sum_{y\in\mathcal S}P(y\mid x,a)\phi_k(y)
  \right]
  =
  \sum_{x\in\mathcal S}c(x)\phi_k(x),
  \;\; \forall k=1,\ldots,K.
\end{aligned}
\end{equation}
The dual variable $\lambda(x,a)$ represents the weight assigned to the
\ac{ALP} constraint associated with the state-action pair $(x,a)$.
Because $\phi_1(x)=1$ for all $x\in\mathcal S$ and
$\sum_{x\in\mathcal S}c(x)=1$, the dual constraint corresponding to
$\phi_1$ implies that every feasible dual solution satisfies
$\sum_{x\in\mathcal S}
\sum_{a\in\mathcal A(x)}
\lambda(x,a)
=
\frac{1}{1-\alpha}.$ Thus, $\blambda$ has the same total mass as an unnormalized discounted state-action occupancy measure. Following \citet[Definition~6.1.3]{guestrin2003planning}, we construct a randomized policy associated with a dual solution. Specifically, at every state with positive dual mass, the probability of selecting an action is obtained by normalizing the corresponding dual variables:
\begin{equation}
\label{eq:dual_fluid_policy}
\pi_{\blambda}(a\mid x)
=
\frac{\lambda(x,a)}
{\sum_{a'\in\mathcal A(x)}\lambda(x,a')},
\;\; \forall a\in\mathcal A(x),
\end{equation}
provided that $\sum_{a'\in\mathcal A(x)}\lambda(x,a')>0$.  We note that, for an optimal primal-dual pair $(\br^*,\blambda^*)$,
complementary slackness between
\eqref{eq:alp_primal_expanded} and
\eqref{eq:alp_dual_expanded} implies that
$\lambda^*(x,a')>0$ for some $a' \in \mathcal A(x)$ only if the corresponding primal constraint is
binding. Hence, $a'$ is a greedy action in
\eqref{eq:greedy_policy_from_alp} for $x$.

If, instead, $\sum_{a'\in\mathcal A(x)}\lambda(x,a')=0$, then $\blambda$ provides no action information at state $x$, and any probability distribution over $\mathcal A(x)$ may be used. To establish convergence and design a practical approach, we adopt a smooth, temperature-controlled policy parameterization from reinforcement learning \citep{papini2022smoothing} that applies a softmax transformation to the dual variables. We then blend the resulting distribution with a prespecified target policy $q(\cdot \,|\, x)$, following the general principle of mixture-policy updates \citep{kakade2002approximately}. The target policy may, for example, select actions uniformly at random, or incorporate domain-specific knowledge.

Specifically, we define the target-guided softmax policy in two steps.
First, among all optimal solutions to \eqref{eq:alp_dual_expanded}, we
select the unique minimum-norm solution
\begin{equation}
  \label{eq:min_norm_dual}
  \overline{\blambda}
  =
  \arg\min_{\blambda^*}
  \left\{
  \frac{1}{2}\lVert\blambda^*\rVert_2^2
  :
  \blambda^*
  \text{ is feasible and optimal for \eqref{eq:alp_dual_expanded}}
  \right\}.
\end{equation}
This solution can be computed by first solving
\eqref{eq:alp_dual_expanded} to obtain its optimal objective value, and
then solving the strictly convex quadratic program
\eqref{eq:min_norm_dual}, subject to the dual feasibility constraints
and the requirement that the dual objective equal that optimal value. Second, we transform $\overline{\blambda}$ into a randomized policy. For
a temperature parameter $\tau>0$, a precision $\epsilon>0$, and a prespecified target policy $q(\cdot \,|\, x)$, we define the randomized policy for each $x \in \mathcal S$ and $a \in \mathcal A(x)$ as
\begin{equation}
\label{eq:target_guided_softmax_policy}
\begin{aligned}
\pi_{\tau,\epsilon,q}(a\mid x;\overline{\blambda})
&=
\frac{\sum_{b\in\mathcal A(x)}\overline{\lambda}(x,b)}
     {\sum_{b\in\mathcal A(x)}\overline{\lambda}(x,b)+\epsilon}
\frac{\exp\!\left(\overline{\lambda}(x,a)/\tau\right)}
     {\sum_{a'\in\mathcal A(x)}\exp\!\left(\overline{\lambda}(x,a')/\tau\right)}
+
\frac{\epsilon}
     {\sum_{b\in\mathcal A(x)}\overline{\lambda}(x,b)+\epsilon}
q(a\mid x).
\end{aligned}
\end{equation}
The temperature parameter $\tau$ controls the sensitivity of the action probabilities to the dual variables: larger values produce a smoother distribution, whereas smaller values place greater probability on actions with larger dual variables. The parameter $\epsilon$ controls the influence of the target policy relative to the total dual mass at state $x$. In particular, if $\sum_{b\in\mathcal A(x)}\overline{\lambda}(x,b)=0$, then \eqref{eq:target_guided_softmax_policy} reduces to $q(a\mid x)$. As the sum of duals increases relative to $\epsilon$, the policy places greater weight on the softmax distribution.

\smallskip
\noindent 
\textit{Weight-Update Rule.}
At iteration $m$, let $\overline{\blambda}^{(m)}$ be the minimum-norm dual solution obtained from \eqref{eq:min_norm_dual} using the current
weights $\bc^{(m)}$, and let $\pi^{(m)}=\pi_{\tau,\epsilon,q}(\cdot\mid\cdot; \overline{\blambda}^{(m)})$ be the corresponding policy from \eqref{eq:target_guided_softmax_policy}. We compute its normalized discounted occupancy $\bmu^{(m)}$, whose component at state $x$ is
$\mu^{(m)}(x)=\mu_{\pi^{(m)},\bnu}(x)$ as defined in \eqref{eq:discounted_occupancy}, and update the state-relevance weights
according to
\begin{equation}
\label{eq:dia_weight_update}
\bc^{(m+1)}
=
(1-\eta)\bc^{(m)}+\eta\bmu^{(m)},
\qquad \eta\in(0,1].
\end{equation}
\edits{Since $\bc^{(m)}$ and $\bmu^{(m)}$ are probability vectors,
\eqref{eq:dia_weight_update} ensures that $\bc^{(m+1)}$ is also a probability vector. Setting $\eta=1$ fully replaces the current weights with the new occupancy measure, while $\eta<1$ dampens the update between iterations. At a fixed point, $\bc^{(m)}=\bmu^{(m)}$, implying that the state-relevance weights equal the normalized discounted occupancy of the policy they induce.}

\section{Convergence and Performance Analysis}
\label{sec:convergence_performance_analysis}

In this section, we show the convergence of the proposed approach in Section \ref{sec:convergence} and present a performance analysis in Section \ref{sec:performance}. \edits{For compactness, we express the \ac{ALP} formulation using vector and matrix notation.} Specifically, let $\Delta(\mathcal S) := \{ \bc \colon \bc \ge 0, \sum_{x \in \mathcal S} c(x) = 1 \}$ be the simplex over $\mathcal S$, denoting the feasible set of weights, and $\mathcal N := \{(x, a) \colon x \in \mathcal S, \, a \in \mathcal A(x)\}$ be the set of feasible state-action pairs.  We let $\bg\in\mathbb R^{|\mathcal N|}$ be the vector of one-period costs with component $g_{(x,a)}=g(x,a)$, and denote by $\bPhi\in\mathbb R^{|\mathcal S|\times K}$ the matrix of basis function evaluations per state, with entries $\Phi_{xk}=\phi_k(x)$ and linearly independent columns. Thus, $\bPhi\br\in\mathbb R^{|\mathcal S|}$ is the vector of approximate values across states. 

Define the Bellman-constraint matrix $\bA\in\mathbb R^{|\mathcal N|\times|\mathcal S|}$, whose rows are indexed by admissible state-action pairs and whose columns are indexed by states,
as 
$ 
A_{(x,a),y} =
\mathbf 1\{y=x\}-\alpha P(y\mid x,a), \forall y\in\mathcal S
$. The component of $\bA\bPhi\br$ associated with $(x,a)$ is
$
(\bA\bPhi\br)_{(x,a)}
=
  (\bPhi\br)(x)
  -
  \alpha\sum_{y\in\mathcal S}
  P(y\mid x,a)(\bPhi\br)(y).
$
Thus, 
$\bc^\top\bPhi\br$  
and 
$\bA\bPhi\br\leq\bg$ 
represent the objective function and the state-action
constraints in \eqref{eq:alp_primal_expanded}, respectively. 
Similarly, recall that $\blambda\in\mathbb R_+^{|\mathcal N|}$ denotes the vector of dual variables. Then,
$\bg^\top\blambda$ 
and 
$(\bA\bPhi)^\top\blambda=\bPhi^\top\bc$ are the objective function and constraints of \eqref{eq:alp_dual_expanded}, respectively.

\subsection{Convergence} 
\label{sec:convergence}

The convergence analysis relies on two supporting Lipschitz-stability results: one for the selected dual solution and another for the resulting stochastic policy. We first establish in Lemma \ref{lemma:continuous_dual_selection_v2} that the selected dual solution is Lipschitz continuous in the state-relevance weights, ruling out abrupt changes in the dual variables caused solely by small perturbations in the weights.

\begin{lemma}[Dual Stability.]
\label{lemma:continuous_dual_selection_v2}

Let $\overline{\blambda}(\bc)$ be the optimal dual vector with respect to state-relevance weights $\bc$ calculated via \eqref{eq:min_norm_dual}. Then, there exists a constant $0<L_{d}<+\infty$ such that
\[
\lVert \overline{\blambda}(\bc)-\overline{\blambda}(\bc') \rVert_1
\le 
L_{d} \, \lVert \bc-\bc' \rVert_1,
\;\; \forall \bc,\bc'\in \Delta(\mathcal{S}).
\]
\end{lemma}
\begin{proof}{Proof of Lemma \ref{lemma:continuous_dual_selection_v2}.}
For any
$\bc\in\Delta(\mathcal S)$, the dual feasible set of \eqref{eq:alp_dual_expanded} is nonempty and
compact: feasibility follows from the discounted state-action occupancy
of any stationary policy with initial distribution $\bc$, and
compactness follows from
$\sum_{(x,a)\in\mathcal N}\lambda(x,a)=1/(1-\alpha)$. Hence, the optimal set is
nonempty and compact, the optimal value function 
$
z(\bc) := \min_{\blambda \ge \bzero} 
  \left \{ 
    \bg^\top\blambda \colon (\bA\bPhi)^\top\blambda=\bPhi^\top\bc 
  \right \}
$
is well defined, and the minimum-Euclidean-norm element
$\overline{\blambda}(\bc)$ exists and is unique. Moreover, the linear program \eqref{eq:alp_dual_expanded} has finitely many bases and only its right-hand side depends affinely on $\bc$. Thus,
its optimal value is a continuous piecewise-affine function of $\bc$. 

\edits{Partition $\Delta(\mathcal S)$ into finitely many polyhedral regions $\Delta_1(\mathcal S),\ldots,\Delta_M(\mathcal S)$ with $M<\infty$, such that $z(\bc)$ is affine on each region. Recall that from \eqref{eq:min_norm_dual}, we have that }
\begin{equation*}
\begin{aligned}
\overline{\blambda}(\bc) 
= 
\argmin_{\blambda} 
\left\{ 
  \frac{1}{2}\lVert\blambda\rVert_2^2
  \colon
  (\bA\bPhi)^\top\blambda=\bPhi^\top\bc, \,
  \blambda\geq\bzero,
  \,
  \bg^\top\blambda=z(\bc)
\right\}.
\end{aligned}
\end{equation*}
On each region $\Delta_i(\mathcal S)$, the function $z(\bc)$ is affine. Restricted to that region, the problem above is therefore a strictly convex multiparametric quadratic program with affine dependence on $\bc$. Its unique optimizer is continuous and piecewise affine on $\Delta_i(\mathcal S)$ \citep{bemporad2002explicit}. At boundaries shared by two regions, the corresponding solutions agree because both solve the same minimum-norm problem and its optimizer is unique. Thus,
$\overline{\blambda}(\bc)$ is continuous and piecewise affine on the
polytope $\Delta(\mathcal S)$ and is therefore globally Lipschitz.
Hence, a constant satisfying the stated inequality exists.
\hfill $\blacksquare$
\end{proof}

\smallskip
However, dual stability alone does not guarantee stable weight updates in \ac{DIA}, because the dual solution is converted into a policy. The next lemma shows that the target-guided softmax construction controls the sensitivity of this policy via $\tau$ and $\epsilon$ for any input duals. In particular, we write $\blambda_x:=\bigl(\lambda(x,a)\bigr)_{a\in\mathcal A(x)}$ for the subvector of dual variables associated with state $x \in \mathcal S$.
\begin{lemma}[Policy Stability.]
\label{lemma:L_Lipschitz_policy}
There exist parameters $\epsilon, \tau > 0$ such that, for any two dual feasible solutions $\blambda, \blambda'$ to \eqref{eq:alp_dual_expanded} and a prespecified target policy $q$,
\[
\lVert \pi_{\tau,\epsilon,q}(\cdot \mid x;\blambda) - \pi_{\tau,\epsilon,q}(\cdot\mid x; \blambda') \rVert_1\le \, L_p \lVert \blambda_x - \blambda'_x \rVert_1, 
\;\; \forall x \in\mathcal S, 
\]
for some constant $L_p < (1-\alpha)/(L_d\alpha)$. Specifically, it suffices to pick any $\epsilon, \tau > 0$ such that  
\begin{equation}
\label{eq:epsilon_tau_conditions}
  \epsilon
  >
  \frac{2\alpha L_d}{1-\alpha},
  \qquad
  \tau
  >
  \frac{\alpha L_d\epsilon}
  {2\left[(1-\alpha)\epsilon-2\alpha L_d\right]}.
\end{equation}
\end{lemma}
\begin{proof}{Proof of Lemma \ref{lemma:L_Lipschitz_policy}.}
Let $x\in\mathcal S$ and $q_x=\bigl(q(a\mid x)\bigr)_{a\in\mathcal A(x)}$. For any $\blambda_x \ge 0$, define
\begin{equation*}
\sigma_\tau(a;\blambda_x)
=
\frac{\exp\!\left(\lambda(x,a)/\tau\right)}
{\sum_{b\in\mathcal A(x)}\exp\!\left(\lambda(x,b)/\tau\right)},
\quad
s(\blambda_x)=\sum_{a\in\mathcal A(x)}\lambda(x,a),
\quad
\omega_\epsilon(\blambda_x)=\frac{s(\blambda_x)}{s(\blambda_x)+\epsilon}.
\end{equation*}
Then, the smoothed policy \eqref{eq:target_guided_softmax_policy} becomes
$\pi_{\tau,\epsilon,q}(\cdot\mid x;\blambda)
=\omega_\epsilon(\blambda_x)\sigma_\tau(\blambda_x)
+[1-\omega_\epsilon(\blambda_x)]q_x$.

For any two actions $a,b\in\mathcal A(x)$, the derivative of the softmax term yields
\begin{equation*}
\frac{\partial\sigma_\tau(a;\blambda_x)}
{\partial\lambda(x,b)}
=
\frac{1}{\tau}\sigma_\tau(a;\blambda_x)
\bigl[\mathbf 1\{a=b\}-\sigma_\tau(b;\blambda_x)\bigr].
\end{equation*}
For each column $b$ of the Jacobian,
\begin{equation*}
\sum_{a\in\mathcal A(x)}
\left|
\frac{\partial\sigma_\tau(a;\blambda_x)}
{\partial\lambda(x,b)}
\right|
=
\frac{2}{\tau}\sigma_\tau(b;\blambda_x)
\bigl[1-\sigma_\tau(b;\blambda_x)\bigr]
\leq
\frac{1}{2\tau}.
\end{equation*}
The equality uses
$\sum_{a\neq b}\sigma_\tau(a;\blambda_x)
=1-\sigma_\tau(b;\blambda_x)$, and the inequality follows from
$p(1-p)\leq1/4$ for any scalar $p\in[0,1]$. Hence, the induced matrix $1$-norm of
the softmax Jacobian is at most $1/(2\tau)$. Applying the integral form
of the mean-value theorem along the line segment between $\blambda_x$
and $\blambda'_x$ therefore gives
$\|\sigma_\tau(\blambda_x)-\sigma_\tau(\blambda'_x)\|_1
\leq\|\blambda_x-\blambda'_x\|_1/(2\tau)$.

Similarly, the derivative of the ratio $s/(s+\epsilon)$ for any scalar $s\geq0$ is bounded by
$1/\epsilon$. Since
$|s(\blambda_x)-s(\blambda'_x)|
\leq\|\blambda_x-\blambda'_x\|_1$, it follows that
$|\omega_\epsilon(\blambda_x)-\omega_\epsilon(\blambda'_x)|
\leq\|\blambda_x-\blambda'_x\|_1/\epsilon$.

Let $\omega=\omega_\epsilon(\blambda_x)$,
$\omega'=\omega_\epsilon(\blambda'_x)$,
$\sigma=\sigma_\tau(\blambda_x)$, and
$\sigma'=\sigma_\tau(\blambda'_x)$. Since
$\pi(\blambda_x)-\pi(\blambda'_x)
=\omega(\sigma-\sigma')+(\omega-\omega')(\sigma'-q_x)$,
$0\leq\omega\leq1$, and $\|\sigma'-q_x\|_1\leq2$, the preceding
bounds yield
\begin{equation*}
\left\|
\pi_{\tau,\epsilon,q}(\cdot\mid x;\blambda)
-\pi_{\tau,\epsilon,q}(\cdot\mid x;\blambda')
\right\|_1
\leq
\left(\frac{1}{2\tau}+\frac{2}{\epsilon}\right)
\left\|\blambda_x-\blambda'_x\right\|_1.
\end{equation*}
Thus, we may take $L_p=1/(2\tau)+2/\epsilon$. Under
\eqref{eq:epsilon_tau_conditions},
$1/(2\tau)< (1-\alpha)/(\alpha L_d)-2/\epsilon$, and hence
$L_p<(1-\alpha)/(\alpha L_d)$, as required.
$\hfill \blacksquare$
\end{proof}

We next show our main results. First, our primary objective is to establish weights $\bc$ that agree with the normalized discounted occupancy of the policy generated by that vector. For each $\bc\in\Delta(\mathcal S)$, let $\overline{\blambda}(\bc)$ be the minimum-norm dual solution from \eqref{eq:min_norm_dual}, and let $\pi_{\bc}:=\pi_{\tau,\epsilon,q}(\cdot\mid\cdot;\overline{\blambda}(\bc))$ be the corresponding policy with $\epsilon, \tau$ satisfying Lemma \ref{lemma:L_Lipschitz_policy}. Define the mapping $F:\Delta(\mathcal S)\to\Delta(\mathcal S)$ by 
$
F(\bc)
=
\bmu_{\pi_{\bc},\bnu},
$ 
where $\bmu_{\pi_{\bc},\bnu}$ is the normalized discounted occupancy measure as in \eqref{eq:discounted_occupancy}.
Equivalently,
\begin{equation}
F(\bc)^\top
=
(1-\alpha)\bnu^\top(I-\alpha P_{\pi_{\bc}})^{-1},
\end{equation}
where $P_{\pi_{\bc}}$ denotes the transition matrix induced by $\pi_{\bc}$, with entries $P_{\pi_{\bc}}(x,y)=\sum_{a\in\mathcal A(x)}\pi_{\bc}(a\mid x)P(y\mid x,a)$. 
The following theorem shows that at least one such vector exists, so the target of the \ac{DIA} iteration is well defined.

\begin{theorem}[Self-Consistent Weights]
\label{thm:self_consistent_weights}
There exists $\bc^\star\in \Delta(\mathcal{S})$ such that $F(\bc^\star)=\bc^\star$.
\end{theorem}
\begin{proof}{Proof of Theorem~\ref{thm:self_consistent_weights}.}
Recall that
$\pi_{\bc}:=\pi_{\tau,\epsilon,q}(\cdot\mid\cdot;
\overline{\blambda}(\bc))$ is the target-guided softmax policy in
\eqref{eq:target_guided_softmax_policy} constructed from the
minimum-norm dual solution associated with $\bc$.
Lemma~\ref{lemma:continuous_dual_selection_v2} implies that
$\overline{\blambda}(\bc)$ is continuous in $\bc$. Lemma~\ref{lemma:L_Lipschitz_policy}
and the definition of $\pi_{\bc}$ then imply that, for every
$x\in\mathcal S$ and $a\in\mathcal A(x)$, the action probability
$\pi_{\bc}(a\mid x)$ is continuous in $\bc$. 

Because each entry of
$P_{\pi_{\bc}}$ is a finite linear combination of the action
probabilities, every entry of $P_{\pi_{\bc}}$ is also continuous in
$\bc$. Moreover,
$P_{\pi_{\bc}}$ is stochastic and $\alpha\in(0,1)$, so
$I-\alpha P_{\pi_{\bc}}$ is nonsingular. Continuity of matrix inversion
therefore implies that $F$ is continuous. By \eqref{eq:discounted_occupancy}, $F(\bc)$ is a normalized discounted
occupancy vector and hence belongs to $\Delta(\mathcal S)$ for every
$\bc\in\Delta(\mathcal S)$. Thus, $F$ is continuous and maps the
nonempty, compact, and convex set $\Delta(\mathcal S)$ into itself. Brouwer's
fixed-point theorem therefore guarantees a vector
$\bc^\star\in\Delta(\mathcal S)$ satisfying
$F(\bc^\star)=\bc^\star$.
$\hfill\blacksquare$
\end{proof}

\smallskip

\edits{Existence alone does not guarantee a unique fixed point or that successive weight updates will converge. By combining dual and policy stability, Theorem~\ref{thm:convergence_v2} establishes that the self-consistent weight vector is unique and that \ac{DIA} converges to it geometrically from any initial weights.}

\begin{theorem}[Convergence]
\label{thm:convergence_v2}
For any step size $\eta\in(0,1]$, the update rule
\begin{equation}
\label{eq:contraction_modulus}
T_\eta(\bc) := (1-\eta)\bc+\eta F(\bc)
\end{equation}
is a contraction on $\Delta(\mathcal S)$ in the $1$-norm, with modulus $\gamma_\eta:=1-\eta+\eta\alpha L_pL_d/(1-\alpha)<1$. Thus, $F$ has a unique fixed point $\bc^\star$, and the \ac{DIA} iterates satisfy
\begin{equation*}
\|\bc^{(m)}-\bc^\star\|_1
\leq
\gamma_\eta^m\|\bc^{(0)}-\bc^\star\|_1,
\qquad m=0,1,2,\ldots.
\end{equation*}
\end{theorem}
\begin{proof}{Proof of Theorem~\ref{thm:convergence_v2}.}
Fix $\bc^1,\bc^2\in\Delta(\mathcal S)$. For $i\in\{1,2\}$, let
$\overline{\blambda}^{,i}:=\overline{\blambda}(\bc^i)$,
$\pi_i:=\pi_{\bc^i}$, $P_i:=P_{\pi_i}$, and
$F_i:=F(\bc^i)$. We use the matrix norm
$\|M\|_\infty:=\max_x\sum_y|M(x,y)|$. For each state $x$, the
definition of $P_i$ and the triangle inequality give
$
\|P_1(x,\cdot)-P_2(x,\cdot)\|_1
\leq
\|\pi_1(\cdot\mid x)-\pi_2(\cdot\mid x)\|_1 
\leq
L_p\|\overline{\blambda}^{,1}_x-
\overline{\blambda}^{,2}_x\|_1.
$
Applying Lemma~\ref{lemma:continuous_dual_selection_v2} therefore yields
\begin{equation}
\label{eq:structural_transition_bound}
\|P_1-P_2\|_\infty
\leq
L_pL_d\|\bc^1-\bc^2\|_1.
\end{equation}

The resolvent identity and the definition of $F$ give
\begin{equation}
\label{eq:F_difference_expansion}
(F_1-F_2)^\top
=
\alpha F_1^\top(P_1-P_2)(I-\alpha P_2)^{-1}.
\end{equation}
Because $F_1$ is a probability vector,
$\|F_1^\top(P_1-P_2)\|_1\leq\|P_1-P_2\|_\infty$. Moreover, $P_2$ is
stochastic, so
$(I-\alpha P_2)^{-1}=\sum_{t=0}^\infty\alpha^tP_2^t$ is nonnegative and
has matrix norm $\|(I-\alpha P_2)^{-1}\|_\infty=1/(1-\alpha)$.
Combining these observations with
\eqref{eq:structural_transition_bound}--\eqref{eq:F_difference_expansion}
gives
\begin{equation}
\label{eq:compact_contraction_chain}
\|F_1-F_2\|_1
\leq
\frac{\alpha}{1-\alpha}\|P_1-P_2\|_\infty
\leq
\frac{\alpha L_pL_d}{1-\alpha}\|\bc^1-\bc^2\|_1.
\end{equation}

Define $\kappa:=\alpha L_pL_d/(1-\alpha)$. The condition
$L_p<(1-\alpha)/(\alpha L_d)$ from
Lemma~\ref{lemma:L_Lipschitz_policy} implies that $\kappa<1$.
It follows that
\begin{equation*}
\|T_\eta(\bc^1)-T_\eta(\bc^2)\|_1
\leq
\bigl(1-\eta+\eta\kappa\bigr)
\|\bc^1-\bc^2\|_1
=
\gamma_\eta\|\bc^1-\bc^2\|_1.
\end{equation*}
Since $\kappa<1$ and $\eta\in(0,1]$, we have $\gamma_\eta<1$; hence,
$T_\eta$ is a contraction. Because $\bc$ and $F(\bc)$ belong to
$\Delta(\mathcal S)$, their convex combination $T_\eta(\bc)$ also
belongs to $\Delta(\mathcal S)$. The simplex is closed and therefore
complete under the $1$-norm, so the Banach fixed-point theorem gives a
unique fixed point and geometric convergence from every initial weight
vector. Finally, because $\eta>0$, $T_\eta(\bc)=\bc$ if
and only if $F(\bc)=\bc$, so this is also the unique fixed point of $F$.
$\hfill\blacksquare$
\end{proof}

\begin{remark}[Generalizations.]
\edits{The preceding results also highlight conditions needed to generalize the policy-generation and weight-update rules beyond \ac{DIA}. Specifically, the target-guided softmax policy can be replaced by any rule $U$ satisfying, uniformly over
states, $\|U(\blambda)(\cdot\mid x)-U(\blambda')(\cdot\mid x)\|_1
\leq L_U\|\blambda_x-\blambda'_x\|_1$ for some constant $L_U$. Because this condition also
implies continuity, the existence result in Theorem~\ref{thm:self_consistent_weights} remains valid. For the
convex-combination update in \eqref{eq:contraction_modulus}, the proof of Theorem~\ref{thm:convergence_v2} remains valid whenever
$L_F:=\alpha L_UL_d/(1-\alpha)<1$. }

\edits{More generally, consider an update $U(\bc,\bmu)$ that maps pairs of probability vectors into $\Delta(\mathcal S)$ satisfying
$\|U(\bc,\bmu)-U(\bc',\bmu')\|_1 \leq L_c\|\bc-\bc'\|_1+L_\mu\|\bmu-\bmu'\|_1$. The resulting iteration is a contraction if $L_c+L_\mu L_F<1$. 
To preserve the self-consistency interpretation of Theorem~\ref{thm:self_consistent_weights}, we require $U(\bc,\bmu)=\bc$ if and only if $\bc=\bmu$. Our update is a special case with $L_c=1-\eta$, $L_\mu=\eta$, and $U=\pi_{\tau,\epsilon,q}$. \hfill $\square$}
\end{remark}

\subsection{Performance Analysis} 
\label{sec:performance}
\edits{The convergence results establish that \ac{DIA} identifies a unique self-consistent weight vector, but they do not indicate how well the resulting greedy policy performs relative to the \ac{ALP} approximation. We therefore derive an a posteriori performance certificate for that policy.} The analysis first bounds policy loss using the Bellman residual and then decomposes this bound into three components: the residual weighted by $\bc$, the mismatch between $\bc$ and the discounted occupancy of the \ac{DIA} policy, and the disagreement between the \ac{DIA} policy and the greedy policy.

\edits{Define $\bJ^*:=(J^*(x))_{x\in\mathcal S}$ and 
let $\bJ^\pi$ be the value vector associated with the stationary randomized policy $\pi$. Given a coefficient vector $\br\in\mathbb R^K$, we can write $\bJ':=\bPhi\br$ for the corresponding value-function approximation and define its state-action Bellman residual by}
$
s_{\bJ'}(x,a)
:=
g(x,a)+\alpha\sum_{y\in\mathcal S}P(y\mid x,a)J'(y)
-J'(x).
$
For a policy $\pi$, let
$g_\pi(x):=\sum_{a\in\mathcal A(x)}\pi(a\mid x)g(x,a)$ and define
the corresponding residual vector by
$
\boldsymbol s_{\bJ'}^\pi
:=
\bg_\pi+\alpha P_\pi\bJ'-\bJ',
$
with
$
s_{\bJ'}^\pi(x)
=
\sum_{a\in\mathcal A(x)}\pi(a\mid x)
s_{\bJ'}(x,a).
$
\begin{theorem}[Residual Loss Certificate]
\label{thm:residual_certificate}
Let $\br$ be feasible for the \ac{ALP}, and set $\bJ'=\bPhi\br$.
Then, for any stationary randomized policy $\pi$,
$\boldsymbol s_{\bJ'}^\pi\geq\bzero$ and
\begin{equation}
\label{eq:residual_loss_certificate}
0
\leq
\bnu^\top(\bJ^\pi-\bJ^*)
\leq
\frac{1}{1-\alpha}
\bmu_{\pi,\bnu}^\top
\boldsymbol s_{\bJ'}^\pi.
\end{equation}
\end{theorem}

\begin{proof}{Proof of Theorem~\ref{thm:residual_certificate}.}
Feasibility of $\br$ implies, componentwise, that
$\bJ'\leq\bJ^*$ and
$\boldsymbol s_{\bJ'}^\pi\geq\bzero$. Further,
\[
\begin{aligned}
\bJ^\pi-\bJ'
&=\bg_\pi+\alpha P_\pi\bJ^\pi-\bJ'
=\alpha P_\pi(\bJ^\pi-\bJ')
  +\bg_\pi+\alpha P_\pi\bJ'-\bJ'
=\alpha P_\pi(\bJ^\pi-\bJ')
  +\boldsymbol s_{\bJ'}^\pi.
\end{aligned}
\]
Rearranging gives
$\bJ^\pi-\bJ'
=(I-\alpha P_\pi)^{-1}
\boldsymbol s_{\bJ'}^\pi.$
Since $\bJ^*\le\bJ^\pi$ and $\bJ'\le\bJ^*$, it follows that
$0\le\bnu^\top(\bJ^\pi-\bJ^*)
\le\bnu^\top(\bJ^\pi-\bJ')
=\frac{1}{1-\alpha}\bmu_{\pi,\bnu}^\top
\boldsymbol s_{\bJ'}^\pi.
\hfill\blacksquare$
\end{proof}

\smallskip
\edits{Notice that Theorem~\ref{thm:residual_certificate} does not require knowledge of the optimal value function $\bJ^*$. Thus, the certificate can be evaluated using quantities available from the \ac{ALP} solution: a feasible coefficient vector $\br$, its value-function approximation $\bJ'=\bPhi\br$, the corresponding Bellman residual, and the discounted occupancy of the policy being evaluated. These quantities can be computed exactly for smaller models or estimated through simulation for larger problems. We next use this certificate to evaluate the greedy policy derived from the \ac{ALP} solution selected by \ac{DIA}.}

For each $\bc\in\Delta(\mathcal S)$, choose an optimal primal solution
$\br^*(\bc)$ and write
$\bJ'_{\bc}:=\bPhi\br^*(\bc)$. Let
$\overline{\blambda}_{\bc}:=\overline{\blambda}(\bc)$ be the selected
minimum-norm dual solution, and recall from
\eqref{eq:target_guided_softmax_policy} that
$\pi_{\bc}(a\mid x):=
\pi_{\tau,\epsilon,q}(a\mid x;\overline{\blambda}_{\bc})$ is the
corresponding target-guided softmax policy. Because $\br^*(\bc)$ and
$\overline{\blambda}_{\bc}$ are primal-dual optimal, they satisfy
complementary slackness. Let $u_{\bc}$ denote the greedy policy defined
in \eqref{eq:greedy_policy_from_alp}, evaluated using $\bJ'_{\bc}$.
We use $R(\bc):=\|F(\bc)-\bc\|_1$ 
and
$
D(\bc):=
\sum_{x\in\mathcal S}[F(\bc)]_x
\bigl[1-\pi_{\bc}(u_{\bc}(x)\mid x)\bigr]
$
to measure, respectively, the occupancy mismatch of the current weights
and the discounted disagreement between the \ac{DIA} policy and the
greedy policy. Finally, let
$m_{\bc}(x):=\sum_{a\in\mathcal A(x)}
\overline\lambda_{\bc}(x,a)$ denote the dual sum at state $x$.

\begin{corollary}[DIA Greedy-Policy Certificate]
\label{cor:greedy_certificate}
For every $\bc\in\Delta(\mathcal S)$,
\begin{equation}
\label{eq:greedy_certificate}
0\le\bnu^\top(\bJ^{u_{\bc}}-\bJ^*)
\le
\frac{
\bc^\top\boldsymbol s_{\bJ'_{\bc}}^{u_{\bc}}
+
\|\boldsymbol s_{\bJ'_{\bc}}^{u_{\bc}}\|_\infty
\left[
R(\bc)
+\dfrac{2\alpha D(\bc)}{1-\alpha}
\right]}
{1-\alpha}.
\end{equation}
\edits{At a fixed point, $R(\bc)=0$. Further,
$s_{\bJ'_{\bc}}^{u_{\bc}}(x)=0 \text{ when }m_{\bc}(x)>0$.
Thus,
$\bc^\top\boldsymbol s_{\bJ'_{\bc}}^{u_{\bc}} =
\sum_{x:m_{\bc}(x)=0}
c(x)s_{\bJ'_{\bc}}^{u_{\bc}}(x).$}
\end{corollary}

\begin{proof}{Proof of Corollary~\ref{cor:greedy_certificate}.}
Applying Theorem~\ref{thm:residual_certificate} with
$\pi=u_{\bc}$ and adding and subtracting $\bc$ gives
\[
\begin{aligned}
\bnu^\top(\bJ^{u_{\bc}}-\bJ^*)
&\le
\frac{
\bc^\top\boldsymbol s_{\bJ'_{\bc}}^{u_{\bc}}
+
\|\boldsymbol s_{\bJ'_{\bc}}^{u_{\bc}}\|_\infty
\|\bmu_{u_{\bc},\bnu}-\bc\|_1}
{1-\alpha}.
\end{aligned}
\]
The triangle inequality yields
\(
\|\bmu_{u_{\bc},\bnu}-\bc\|_1
\le
R(\bc)
+\|\bmu_{u_{\bc},\bnu}-F(\bc)\|_1.
\)
The occupancy measures satisfy
$
(\bmu_{u_{\bc},\bnu}-F(\bc))^\top
=
\alpha F(\bc)^\top
(P_{u_{\bc}}-P_{\pi_{\bc}})
(I-\alpha P_{u_{\bc}})^{-1}.
$
Since the inverse is nonnegative with row sums $1/(1-\alpha)$ and
\(
\|P_{u_{\bc}}(x,\cdot)-P_{\pi_{\bc}}(x,\cdot)\|_1
\le
2\bigl[1-\pi_{\bc}(u_{\bc}(x)\mid x)\bigr],
\)
we obtain
\[
\|\bmu_{u_{\bc},\bnu}-F(\bc)\|_1
\le\frac{2\alpha D(\bc)}{1-\alpha}.
\]
Combining these bounds gives~\eqref{eq:greedy_certificate}. At a
\ac{DIA} fixed point, $F(\bc)=\bc$, and hence $R(\bc)=0$.
If $m_{\bc}(x)>0$, then
$\overline\lambda_{\bc}(x,a)>0$ for at least one action $a$.
Complementary slackness implies that the corresponding state-action
residual is zero. Because all state-action residuals are nonnegative and
$u_{\bc}$ minimizes them, it follows that
$s_{\bJ'_{\bc}}^{u_{\bc}}(x)=0$. Therefore, only states with
$m_{\bc}(x)=0$ can contribute to
$\bc^\top\boldsymbol s_{\bJ'_{\bc}}^{u_{\bc}}$, which proves
the final identity.
$\hfill\blacksquare$
\end{proof}

\smallskip
The bound in \eqref{eq:greedy_certificate} decomposes policy loss into a
weighted Bellman-residual term and two mismatch terms. The quantity
$\bc^\top\boldsymbol s_{\bJ'_{\bc}}^{u_{\bc}}$ measures the residual of the \ac{ALP} approximation over the states emphasized by $\bc$. It reflects the expressivity of the basis functions and choice of state-relevance weights. The term $R(\bc)$ measures
the mismatch between $\bc$ and the discounted occupancy of $\pi_{\bc}$,
while $D(\bc)$ measures the disagreement between $\pi_{\bc}$ and the
greedy policy $u_{\bc}$. At a \ac{DIA} fixed point,
$R(\bc)=0$, but $D(\bc)$ may be positive because the
fixed point is defined using the occupancy of $\pi_{\bc}$ rather than
the greedy policy $u_{\bc}$. Thus, the certificate is tighter when the Bellman residual is small, $\bc$ closely matches the occupancy of $\pi_{\bc}$, and $\pi_{\bc}$ places most of its probability on greedy actions.


\section{Computational Evaluation}
\label{sec:numerical_experiments}

\edits{We present a numerical study for two common settings: controlled queueing models (Section~\ref{sec:queueing}) and a healthcare scheduling application (Section~\ref{sec:scheduling}). In each, we compare alternative methods for selecting state-relevance weights and evaluate policy performance through simulation. Computational details are provided in Appendix~\ref{sec:numerical_details}. The results show that adaptive weight selection can substantially improve policy quality relative to heuristic weights while matching or improving upon primal-based methods at lower computational cost, particularly when the basis functions are sufficiently expressive.}

\subsection{Queueing System} 
\label{sec:queueing}
\edits{Controlled queueing problems provide standard benchmarks for evaluating state-relevance weights. To this end, we adapt and extend the instances of \citet{de2003linear} and \citet{de2008choosing} in which jobs arrive stochastically, the system state is defined by queue lengths, and the decision maker selects service rates or allocates servers subject to capacity constraints.}

For our analysis, we evaluate ad hoc state-relevance weights from the literature, \ac{PIA}, \ac{DIA}, and an additional Random-policy benchmark. Random selects an admissible action uniformly at random at each state, using the resulting discounted occupancy measure as the state-relevance weights. This benchmark allows us to isolate the value of systematically incorporating information from the simulated or computed policy, beyond the effect of occupancy-based weighting alone. The settings for the queueing experiments are described in detail in Appendix~\ref{app:queue}. Additional insights on discounted occupancy measure and convergence are included in Appendix \ref{app:convergence_insights}.

\paragraph{Single-Dimensional Queue.} We compare the fixed state-relevance weights proposed by \citet{de2003linear}, i.e., $c(x) = (1-\xi)\,\xi^x$ with $\xi = \{0.9,0.999\}$, against the adaptive weights produced by \ac{PIA} and \ac{DIA}. For the adaptive methods, we consider three initial state distributions $\nu(x)$ representing different levels for the initial amount of congestion:
    \begin{enumerate}
  \item \textbf{Geometric:} A low-congestion system favoring $\xi=0.9$ where $\nu(x) = (1-\xi)\,\xi^x$.
  \item \textbf{Empty-System:} An empty system such that $\nu(x) = 1$ for $x=0$ and zero otherwise.
    \item \textbf{Reverse-Geometric:} A high-congestion system favoring $\xi = 0.999$ where $\nu(x) = (1-\xi)\xi^{S-x}$.
\end{enumerate}

\edits{For states in which the dual solution assigns no mass, \ac{DIA} requires a target policy to determine the corresponding actions. We evaluate two variants: \ac{DIA} (Random), which uses a uniformly random target policy, and \ac{DIA} ($\xi=0.9$), which uses the approximate policy generated by the fixed-weight \ac{ALP} with $\xi=0.9$. Table~\ref{tab:policy_comparison_single_queue} reports discounted costs estimated from $10^6$ Monte Carlo episodes of length $T=343$, chosen so that $\alpha^T=0.98^{343}<10^{-3}$. Both \ac{PIA} and \ac{DIA} converge within 10 iterations.}

\begin{table}[ht]
\centering
\footnotesize
\caption{Policy comparison under three initial state distributions for the single-dimensional control queue.}
\begin{tabular}{lccc}
  \toprule
  & \multicolumn{3}{c}{\textbf{Discounted Costs ($\pm$95\% \ac{CI})}} \\
  \cmidrule(lr){2-4}
  \textbf{Policy} & \textbf{Geometric} & \textbf{Empty-System} & \textbf{Reverse-Geometric} \\
  \midrule
  Optimal policy                     & $391.23 \pm 0.76$  & $117.12 \pm 0.09$  & $2,496,721.59 \pm 1.03$ \\
  Fixed weights ($\xi=0.9$)          & $393.81 \pm 0.80$  & $117.12 \pm 0.09$  & $2,497,280.50 \pm 0.67$ \\
  Fixed weights ($\xi=0.999$)        & $517.26 \pm 0.74$  & $228.76 \pm 0.13$  & $2,496,725.38 \pm 1.04$ \\
  Policy induced by \ac{PIA}              & $393.77 \pm 0.79$  & $118.78 \pm 0.11$  & $2,496,725.38 \pm 1.04$ \\
  Policy induced by \ac{DIA} (Random)     & $394.89 \pm 0.81$  & $118.78 \pm 0.11$  & $2,496,725.38 \pm 1.04$ \\
  Policy induced by \ac{DIA} ($\xi=0.9$)  & $393.77 \pm 0.79$  & $118.78 \pm 0.11$  & $2,496,725.38 \pm 1.04$ \\
  Policy induced by Random  & $394.89 \pm 0.81$  & $118.78 \pm 0.11$  & $2,496,725.38 \pm 1.04$ \\
  \bottomrule
\end{tabular}
\label{tab:policy_comparison_single_queue}
\end{table}

Table~\ref{tab:policy_comparison_single_queue} shows that the effectiveness of a fixed weighting scheme depends on the initial state distribution. The fixed-weight policy with $\xi=0.9$, which emphasizes low-congestion states, performs nearly optimally under the \textit{Geometric} and \textit{Empty-System} distributions but performs poorly under \textit{Reverse-Geometric}. Conversely, the fixed-weight policy with $\xi=0.999$, which places greater emphasis on high-congestion states, underperforms under the \textit{Geometric} and \textit{Empty-System} distributions but is nearly optimal under \textit{Reverse-Geometric}. The adaptive approaches remain close to the best fixed-weight benchmark across all three initial distributions. The results also show how a fixed-weight policy can serve as an informative target within \ac{DIA}: when the dual solution assigns little or no mass to a state, the target policy supplies the corresponding action probabilities. Accordingly, \ac{DIA} ($\xi=0.9$) matches \ac{PIA} and performs weakly better than Random and \ac{DIA} (Random).

Figure~\ref{fig:combined_single_queue} compares the approximate and true cost-to-go functions, along with the actions associated with the induced policies. We observe that fixed-weight approaches concentrate approximation accuracy in different regions of the state space. Specifically, $\xi=0.9$ emphasizes low-congestion states, whereas $\xi=0.999$ emphasizes high-congestion states. In contrast, \ac{PIA} and \ac{DIA} (Random) adapt the weights toward states frequently visited under the induced policy, producing the greatest accuracy in regions where the policy is most likely to be concentrated.

\begin{figure}[htbp]
  \centering
  \begin{subfigure}[b]{0.33\textwidth}
    \includegraphics[width=\graphscaler\linewidth]{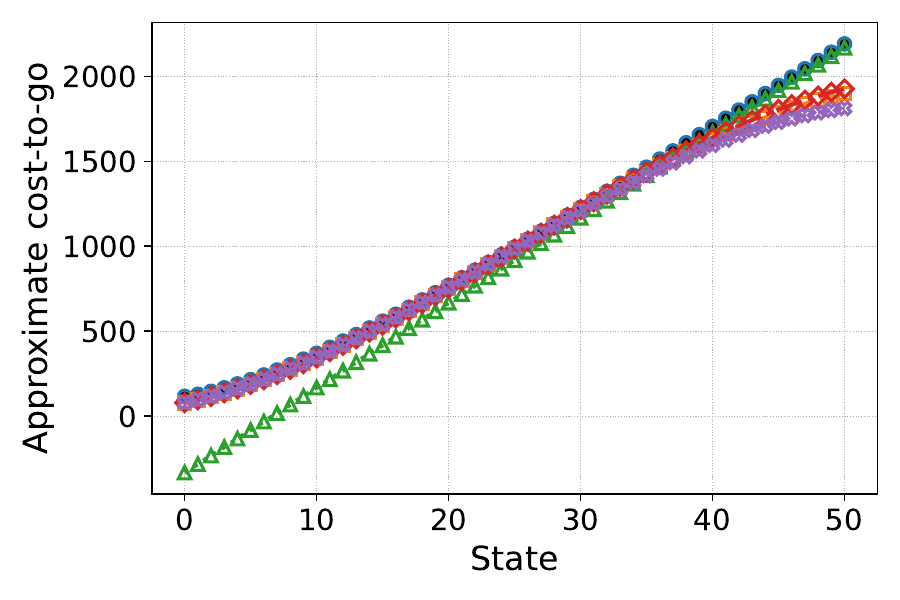}
    \caption{Approx.\ Cost-to-go: Geometric}
    \label{fig:app_cost_to_go_geo}
  \end{subfigure}\hfill
  \begin{subfigure}[b]{0.33\textwidth}
    \includegraphics[width=\graphscaler\linewidth]{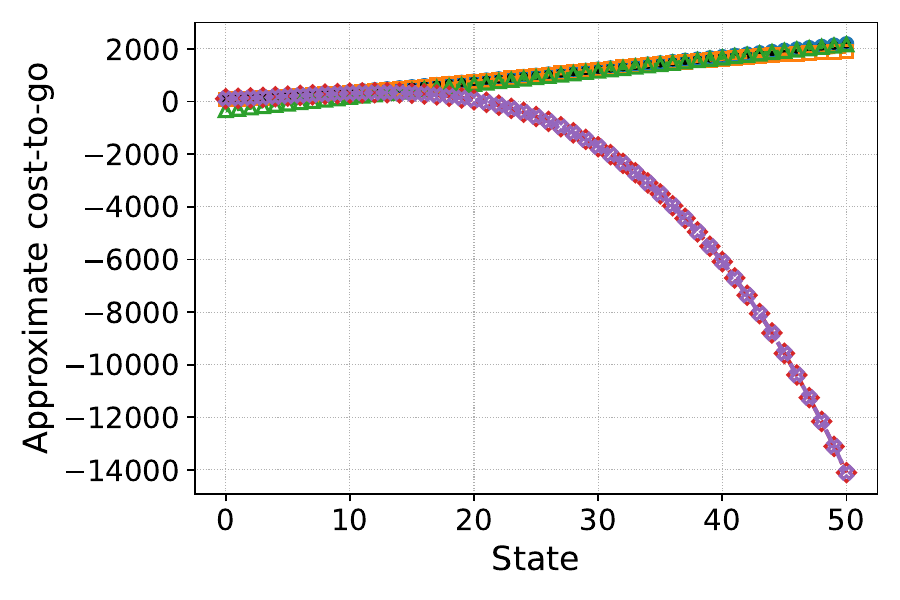}
    \caption{Approx.\ Cost-to-go: Empty}
    \label{fig:app_cost_to_go_empty}
  \end{subfigure}\hfill
  \begin{subfigure}[b]{0.33\textwidth}
    \includegraphics[width=\graphscaler\linewidth]{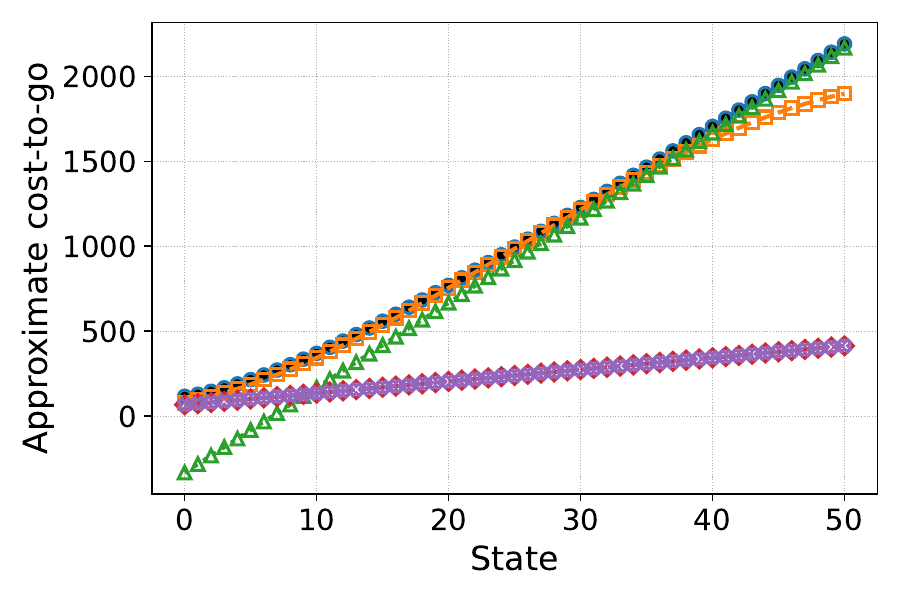}
    \caption{Approx. Cost-to-go: Rev. Geom}
    \label{fig:app_cost_to_go_rev}
  \end{subfigure}

  \vspace{1em}

  \begin{subfigure}[b]{0.33\textwidth}
    \includegraphics[width=\graphscaler\linewidth]{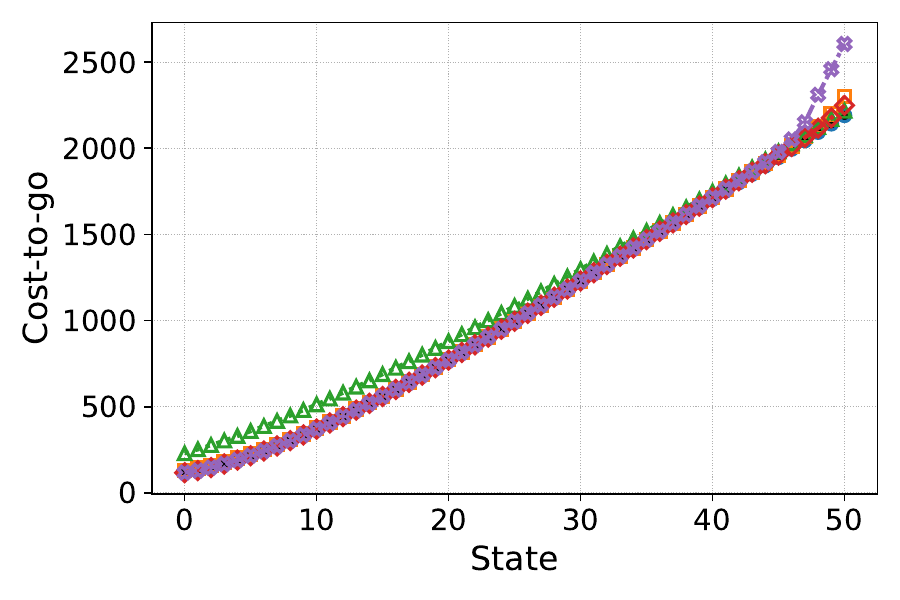}
    \caption{Cost-to-go: Geometric}
    \label{fig:cost_to_go_geo}
  \end{subfigure}\hfill
  \begin{subfigure}[b]{0.33\textwidth}
    \includegraphics[width=\graphscaler\linewidth]{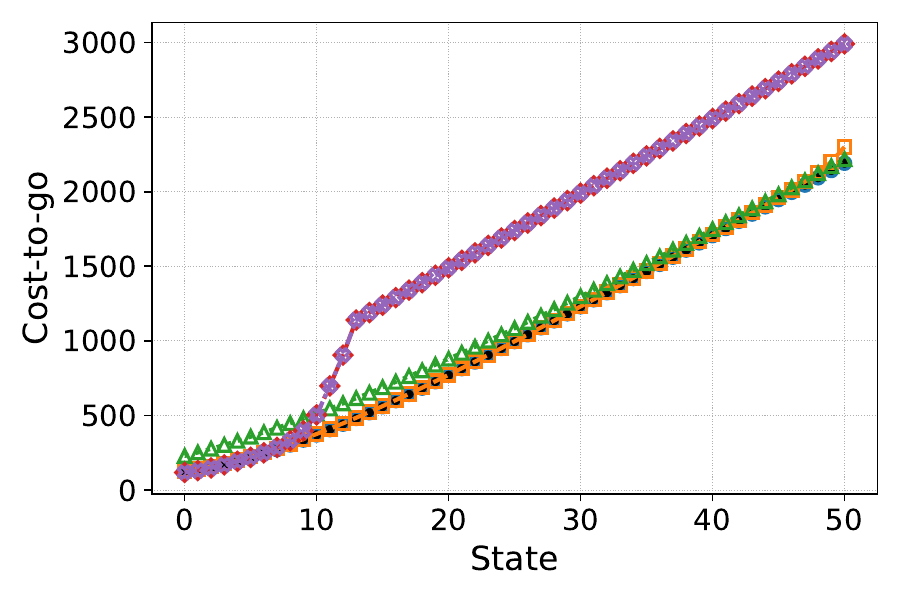}
    \caption{Cost-to-go: Empty}
    \label{fig:cost_to_go_empty}
  \end{subfigure}\hfill
  \begin{subfigure}[b]{0.33\textwidth}
    \includegraphics[width=\graphscaler\linewidth]{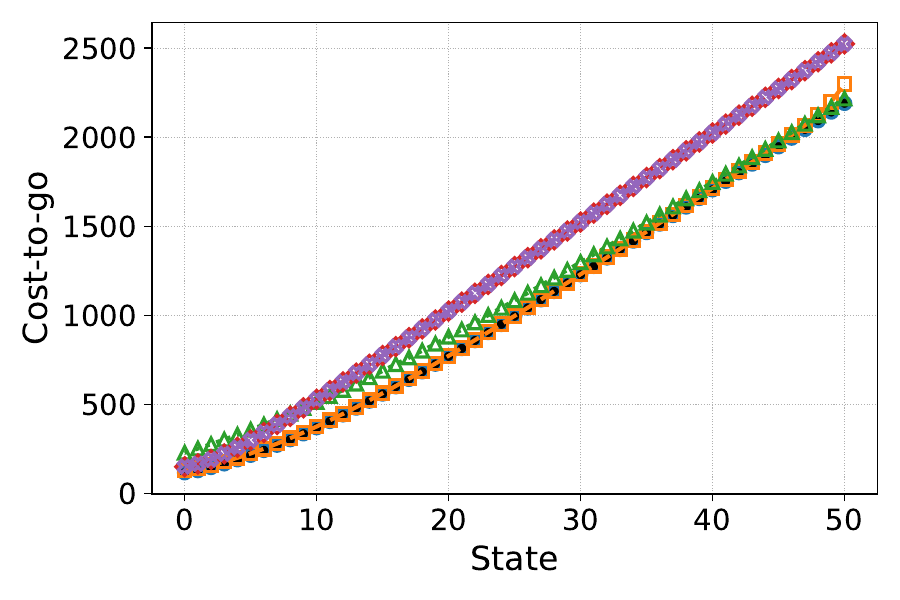}
    \caption{Cost-to-go: Reverse Geometric}
    \label{fig:cost_to_go_rev}
  \end{subfigure}

  \vspace{1em}

  \begin{subfigure}[b]{0.33\textwidth}
    \includegraphics[width=\graphscaler\linewidth]{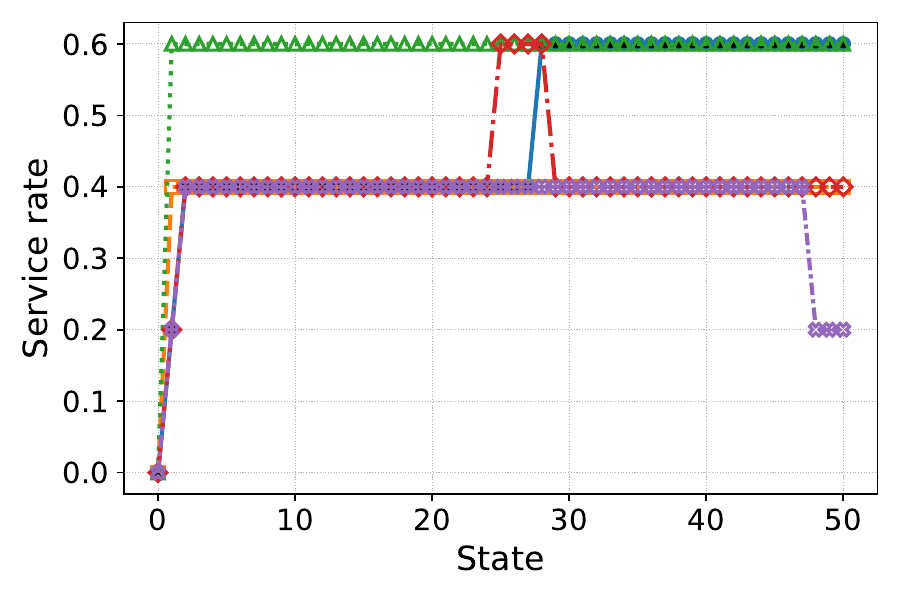}
    \caption{Actions: Geometric}
    \label{fig:action_geo}
  \end{subfigure}\hfill
  \begin{subfigure}[b]{0.33\textwidth}
    \includegraphics[width=\graphscaler\linewidth]{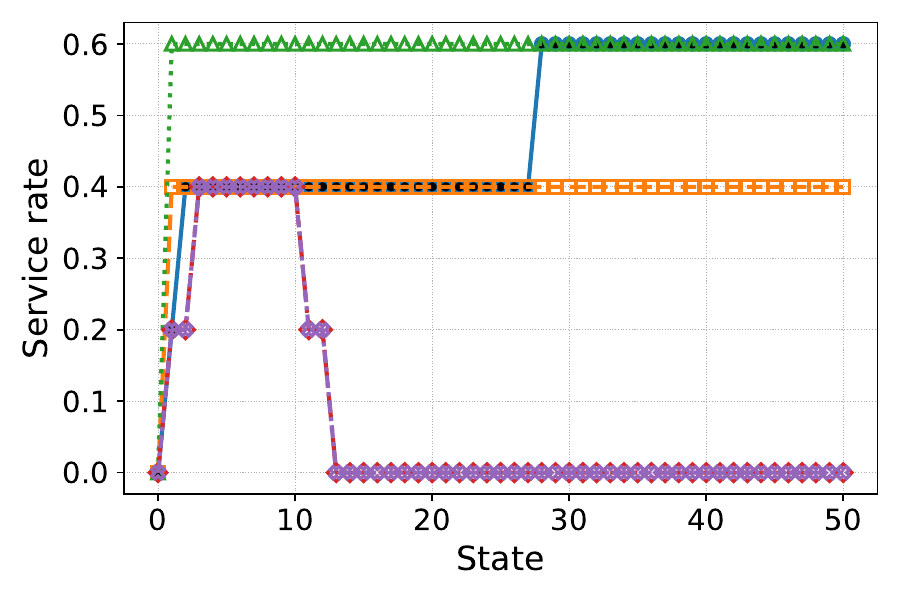}
    \caption{Actions: Empty}
    \label{fig:action_empty}
  \end{subfigure}\hfill
  \begin{subfigure}[b]{0.33\textwidth}
    \includegraphics[width=\graphscaler\linewidth]{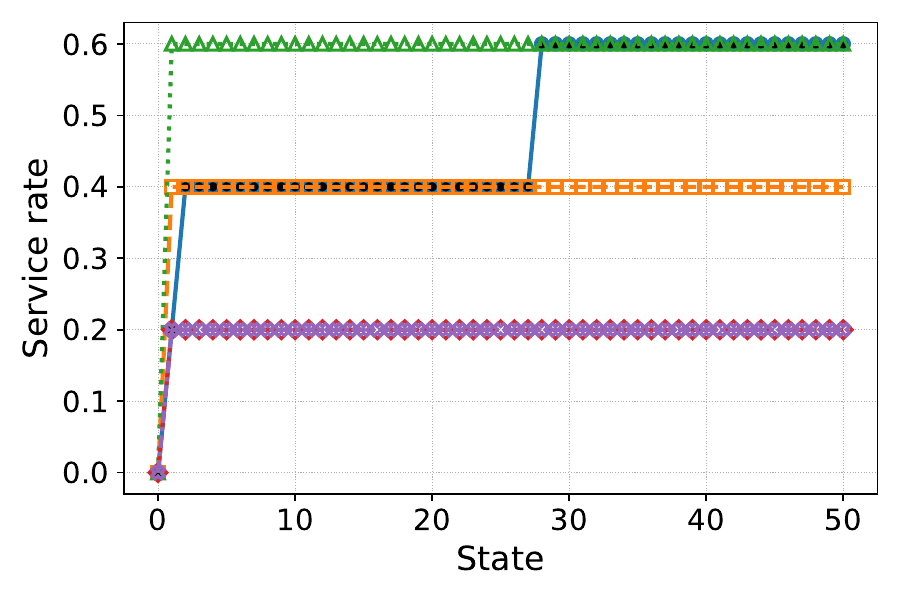}
    \caption{Actions: Reverse Geometric}
    \label{fig:action_rev}
  \end{subfigure}

  \begin{subfigure}[b]{0.8\textwidth}
    \centering
    \includegraphics[width=\graphscaler\linewidth]{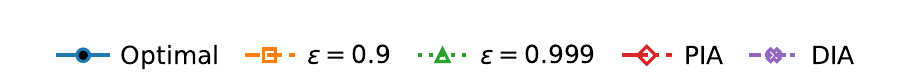}
  \end{subfigure}

  \caption{Comparison under three initial state distributions (Geometric, Empty, Reverse Geometric) where (a--c) is the approximate cost-to-go, (d--f) is the true cost-to-go, and (g--i) are the actions.}
  \label{fig:combined_single_queue}
\end{figure}

\smallskip
\paragraph{Eight-Dimensional Queue.} We next consider an eight-dimensional queueing network \citep{de2003linear, de2008choosing}; see Appendix~\ref{app:8d_queue} for a detailed description. Table~\ref{tab:policy_comparison_8d_queue} reports discounted costs and runtimes, where policy costs are estimated from $40,000$ Monte Carlo episodes of length $T=1{,}379$, chosen so that $\alpha^T = 0.995^{1379} < 10^{-3}$. In the sampled LP constraint set, the maximum total queue length is 193 and the maximum single-queue length is 92.

\begin{table}[ht]
\centering
\footnotesize
\caption{Policy comparison for the eight-dimensional control queue.}
\begin{tabular}{lrr}
  \toprule
  \textbf{Policy} & \textbf{Runtime (s)} & \textbf{Discounted Cost} \textbf{($\pm$95\% \ac{CI})} \\
  \midrule
  Fixed weights ($\xi=0.9$)   & ---      & $13,906.48 \pm 53.08$  \\
  Random baseline              & 1,574.53  & $13,380.16 \pm 50.77$  \\
  Policy induced by \ac{PIA}        & 13,751.61 & $13,458.07 \pm 50.68$  \\
  Policy induced by \ac{DIA} (Random) & 3,660.20  & $13,334.33 \pm 51.25$  \\
  Policy induced by \ac{DIA} ($\xi = 0.9$) & 14,588.35  & $13,195.70 \pm 49.82$   \\
  \bottomrule
\end{tabular}
\label{tab:policy_comparison_8d_queue}
\end{table}

The results show that fixed state-relevance weights are not uniformly reliable. Relative to the fixed-weight baseline ($\xi=0.9$), mean discounted cost decreases by (3.78\%) under the Random baseline, (3.22\%) under \ac{PIA}, (4.11\%) under \ac{DIA} (Random), and (5.11\%) under \ac{DIA} ($\xi=0.9$). All four alternatives have 95\% confidence intervals that are strictly below that of the fixed-weight baseline, indicating statistically significant improvements. \ac{DIA} ($\xi=0.9$) performs the best, with a confidence interval strictly below those of the other non-baseline methods. While \ac{DIA} (Random) outperforms the remaining alternatives, its mean cost is still statistically higher than \ac{DIA} ($\xi = 0.9$).

The computational results reveal a tradeoff between algorithmic runtime and solution quality. In particular, \ac{DIA} (Random) is the fastest adaptive method, requiring 3,660s, as compared with 13,752s for \ac{PIA} and 14,588s for \ac{DIA} ($\xi=0.9$), while still substantially improving upon the fixed-weight baseline. Although \ac{DIA} ($\xi=0.9$) achieves the lowest mean discounted cost, it requires roughly four times the runtime of \ac{DIA} (Random). This difference reflects the algorithmic design: \ac{PIA} and \ac{DIA} ($\xi=0.9$) repeatedly compute greedy actions through the target policy, whereas \ac{DIA} (Random) updates the policy directly from the dual variables. Consequently, \ac{DIA} (Random) offers the best balance between computational efficiency and policy performance, while \ac{DIA} ($\xi=0.9$) provides modest additional performance when greater computational effort is acceptable.

\subsection{Healthcare Scheduling}
\label{sec:scheduling}
\edits{We next study a multi-priority advance-scheduling problem for diagnostic imaging capacity \citep{patrick2008dynamic}. This setting yields a particularly challenging \ac{ALP} because of its large state and action spaces, which requires the use of more sophisticated approaches (e.g., column generation) to solve the underlying dual linear program.}

\smallskip 
\paragraph{Small-Case Scenario.} We first consider a small outpatient-clinic instance with $C_1 = 10$ appointment slots per day and a diversion capacity limit of $C_2 = 4$. Table~\ref{tab:scheduling_small} reports runtimes and discounted costs. The policies induced by fixed weights $\bc=\boldsymbol{\nu}$, Random, and \ac{DIA} (Random) have overlapping 95\% confidence intervals for discounted cost, suggesting that these approaches effectively produce the same \ac{ALP} policy. \ac{PIA} is the exception: its cost is more than $22$ times larger as the algorithm fails to converge, cycling between two poor policies rather than reaching a stable fixed point. Moreover, \ac{PIA} requires at least twice the runtime of \ac{DIA}. Nevertheless, the choice $\bc=\boldsymbol{\nu}$ performs particularly well, achieving a relatively low discounted cost and the shortest runtime by an order of magnitude.

\begin{table}[ht]
\centering
\footnotesize
\begin{tabular}{lrr}
  \toprule
  \textbf{Policy} &\textbf{ Runtime (s)} & \textbf{Discounted Cost ($\pm$95\% \ac{CI})} \\
  \midrule
  Policy induced by fixing $\bc=\boldsymbol{\nu}$   & 0.96 & $6,011.30 \pm 55.94$ \\
  Policy induced by Random   & 10.26 & $5,960.01 \pm 55.59$ \\
  Policy induced by \ac{PIA}            & 310.84 & $135,643.85 \pm 909.17$ \\
  Policy induced by \ac{DIA} (Random)  & 155.54 & $5,980.67 \pm 56.01$ \\
  \bottomrule
\end{tabular}
\caption{Estimated runtime and discounted costs for the small-case scheduling scenario.}
\label{tab:scheduling_small}
\end{table}


\begin{table}[ht]
\centering
\footnotesize
\begin{tabular}{lrr}
  \toprule
  \textbf{Policy} & \textbf{Runtime (s)} & \textbf{Discounted Cost ($\pm$95\% \ac{CI})} \\
  \midrule
  Policy induced by fixing $\bc=\boldsymbol{\nu}$          & 0.93 & $568,112.93 \pm 1,331.77$ \\
  Policy induced by Random   & 10.42 & $10,760.87 \pm 140.38$ \\
  Policy induced by \ac{PIA}            & 308.43 & $10,748.94 \pm 140.42$ \\
  Policy induced by \ac{DIA} (Random)  & 156.00 & $10,788.38 \pm 140.32$ \\
  \bottomrule
\end{tabular}
\caption{Estimated runtimes and discounted costs for the large-case scheduling scenario.}
\label{tab:scheduling_large}
\end{table}

\smallskip
\paragraph{Large-Case Scenario.} \edits{Following \citet{patrick2008dynamic}, we next consider a higher-capacity setting with $C_1=60$ scans per day, where most demand comes from lower-priority patients. Daily demand follows independent Poisson distributions with rates $\lambda=(10,20,30)$ across the three priority classes, while diversion capacity remains limited to $C_2=4$. All other model parameters remain unchanged. We evaluate each policy using $11,000$ Monte Carlo episodes of length $T=688$.}

In this larger instance, fixing $\bc=\boldsymbol{\nu}$ performs poorly. The induced policy cost is roughly 53 times larger than the others, which are statistically indistinguishable due to their overlapping 95\% confidence intervals. Unlike the small-case scenario, however, \ac{PIA} converges to a stable fixed point and its performance matches that of Random and \ac{DIA} (Random). Runtime follows a similar pattern as in the small scheduling example with Random requiring roughly an order of magnitude less computation time than \ac{DIA} (Random). We note that the projected feature moments $\bc^{\top}\bPhi$ are not identical across these three methods. However, the resulting \ac{ALP} coefficients $\br^{*}$ are nearly the same. This similarity reflects degeneracy in the projected dual problem: the dual ALP formulation only has $K_\Phi = 34$ dual constraints relative to a very large state-action space ($|\mathcal{S}| \geq 10^{58}$ states).

\smallskip
\paragraph{Linear vs.\ Higher-Order Basis Functions.}
\edits{The previous experiments for the multi-priority advance-scheduling problem suggest that the policy used to generate the state-relevance weights has limited effect on the resulting \ac{ALP} solution, as the Random policy performs well in both instances. We next examine whether this finding changes when the basis functions provide a richer approximation of the value function. Specifically, for a tiny instance (see Appendix~\ref{app:scheduling}), we increase the polynomial degree from 1 to 4, expanding the number of basis functions from $K_\Phi=5$ to $K_\Phi=70$.}

\edits{Results are reported in Table~\ref{tab:scheduling_full_state}. With the degree-one basis ($K_\Phi=5$), Random and \ac{DIA} (Random) produce the same policy and costs. Although the methods generate different values of $\bc^\top\bPhi$, the limited basis and degeneracy in \eqref{eq:alp_dual_expanded} yield the same dual solution and, consequently, the same $\br^*$ and greedy policy. With the more expressive degree-four basis ($K_\Phi=70$), however, the \ac{ALP} can better exploit the information contained in the state-relevance weights to produce a greedy policy that more closely reflects the states visited under the induced occupancy distribution. As a result, \ac{DIA} (Random) reduces discounted cost relative to Random by $1,490.42$ ($9.01\%$), with non-overlapping confidence intervals. Thus, adaptive weighting becomes more valuable when the basis functions are sufficiently expressive to translate occupancy information into improved policies.}


\begin{table}[ht]
\centering
\footnotesize
\begin{tabular}{lccccc}
  \toprule
  \textbf{Basis} & \textbf{$K_\Phi$} & \textbf{Random ($\pm$95\% \ac{CI})} & \textbf{\ac{DIA} (Random) ($\pm$95\% \ac{CI})} & \textbf{$\Delta$} & \textbf{$\Delta(\%)$} \\
  \midrule
  Degree 1 & 5 & $29{,}001.42 \pm 54.03$ & $29{,}001.42 \pm 54.03$ & $0.00$ & $0.0\%$ \\
  Degree 4 & 70 & $16{,}546.82 \pm 52.28$ & $15{,}056.40 \pm 43.97$ & $1{,}490.42$ & $9.01\%$ \\
  \bottomrule
\end{tabular}
\caption{A comparison of discounted costs between Random and \ac{DIA} (Random) under degree-one (5) and degree-four (70) basis functions. Here $\Delta = J_{\mathrm{Random}} - J_{\text{\ac{DIA} (Random)}}$. }
\label{tab:scheduling_full_state}
\end{table}

These experiments suggest three main findings. First, as in the controlled queueing settings, fixed state-relevance weights can perform well, but they are not robust, and it is generally difficult to determine a priori when they will produce poor policies. Second, \ac{PIA} is computationally expensive and may fail to converge to a fixed point in large-scale settings. Finally, \ac{DIA} (Random) provides an effective balance between solution quality and computational effort. \edits{Although it is more computationally demanding than using the Random policy directly, it can yield substantial performance improvements when the basis functions are sufficiently informative.}

\section{Conclusion}
\label{sec:conclusion}

This paper proposes \ac{DIA}, a dual-based method for adaptively selecting state-relevance weights in infinite-horizon discounted \ac{ALP} problems. \ac{DIA} constructs a smooth stochastic policy from a selected dual \ac{ALP} solution and updates the weights using the policy's normalized discounted occupancy measure. Under Lipschitz conditions on the selected dual solution and the policy-extraction map, the resulting update is a contraction and converges globally at a geometric rate to a unique fixed point that matches the discounted occupancy rate of the induced policy. We also derive an a posteriori performance certificate that separates the effects of the weighted Bellman residual, occupancy mismatch, and disagreement between the stochastic dual-derived policy and the greedy policy.

\edits{The numerical results establish two main findings. First, our dual-based approach matches or improves upon existing primal-based methods at substantially lower computational cost. Second, adaptive state-relevance weights can materially improve policy quality relative to heuristic weighting schemes when the basis functions are sufficiently expressive to exploit occupancy information. Otherwise, approximation error dominates and limits the influence of weight selection. Future research could use policy and occupancy differences to identify where richer basis functions are needed, develop adaptive smoothing schemes that balance convergence stability with information from the dual solution, further integrate dual information into greedy-policy construction, and extend the approach to sampled-constraint and continuous-state \ac{ALP} formulations.}

\bibliographystyle{informs2014}
\bibliography{ref}

\clearpage
\begin{APPENDICES}
\small


\section{Primal Iterative Algorithm}
\label{app:pia_algorithm}

Following \citet{le2004choose}, we combine two results established by \citet{de2003linear}: a greedy-policy performance bound \eqref{eq:pia_greedy_performance} and  an \ac{ALP} approximation bound \eqref{eq:pia_alp_approximation}. Let $\widetilde{\br}(\bc)$ be an optimal solution of \eqref{eq:alp_primal_expanded}, define
\[
\widehat{\bJ}_{\bc}
=
\bPhi\widetilde{\br}(\bc),
\]
and let $u_{\bc}$ be a greedy policy with respect to $\widehat{\bJ}_{\bc}$. Choose $\bv\in\mathbb R^K$ such that $V=\bPhi\bv$ is a Lyapunov
function for the discounted \ac{MDP}. Specifically, $V(x)>0$ for every $x\in\mathcal S$, and its associated Lyapunov constant,
\[
\beta_V
=
\max_{x\in\mathcal S}
\frac{
\alpha\displaystyle\max_{a\in\mathcal A(x)}
\sum_{y\in\mathcal S}P(y\mid x,a)V(y)}
{V(x)},
\]
satisfies $\beta_V<1$. For this Lyapunov function and any nonnegative weight vector $\boldsymbol\rho$, define
\[
\|\boldsymbol h\|_{1,\boldsymbol\rho}
=
\sum_{x\in\mathcal S}\rho(x)|h(x)|,
\qquad
\|\boldsymbol h\|_{\infty,1/V}
=
\max_{x\in\mathcal S}\frac{|h(x)|}{V(x)}.
\]

Because $\widehat{\bJ}_{\bc}$ is feasible for the \ac{ALP}, it satisfies $\widehat{\bJ}_{\bc}\leq\bJ^*$. The first bound relates the performance
of its greedy policy to its value-function approximation error:
\begin{equation}
\label{eq:pia_greedy_performance}
\|\bJ^{u_{\bc}}-\bJ^*\|_{1,\bnu}
\leq
\frac{1}{1-\alpha}
\|\bJ^*-\widehat{\bJ}_{\bc}\|_{1,\bmu_{u_{\bc},\bnu}}.
\end{equation}
Thus, the performance loss of $u_{\bc}$ is controlled by the approximation error in states that are frequently visited under $u_{\bc}$.

The second bound controls the approximation error produced by the \ac{ALP}:
\begin{equation}
\label{eq:pia_alp_approximation}
\|\bJ^*-\widehat{\bJ}_{\bc}\|_{1,\bc}
\leq
\frac{2\,\bc^\top V}{1-\beta_V}
\inf_{\br\in\mathbb R^K}
\|\bJ^*-\bPhi\br\|_{\infty,1/V}.
\end{equation}
This inequality compares the \ac{ALP} solution with the best approximation to $\bJ^*$ available in the span of the basis functions.

The weighting vectors in \eqref{eq:pia_greedy_performance} and \eqref{eq:pia_alp_approximation} coincide if $\bc$ satisfies the self-consistency equation
\begin{equation}
\label{eq:tallec_obj}
\bc^\top
=
\bmu_{u_{\bc},\bnu}^{\top}
=
(1-\alpha)\bnu^\top
\left(I-\alpha P_{u_{\bc}}\right)^{-1}.
\end{equation}
Substituting \eqref{eq:pia_alp_approximation} into \eqref{eq:pia_greedy_performance} then gives
\begin{equation}
\label{eq:tallec_bound}
\|\bJ^{u_{\bc}}-\bJ^*\|_{1,\bnu}
\leq
\frac{2\,\bc^\top V}
{(1-\alpha)(1-\beta_V)}
\inf_{\br\in\mathbb R^K}
\|\bJ^*-\bPhi\br\|_{\infty,1/V}.
\end{equation}
Therefore, self-consistency aligns the state-relevance weights used by the \ac{ALP} with the discounted occupancy measure appearing in the greedy-policy performance bound.

Motivated by \eqref{eq:tallec_obj}, the algorithm proposed by \citet{le2004choose}, which we refer to as the \ac{PIA}, seeks a fixed point of the mapping
\[
\bc
\longmapsto
\bmu_{u_{\bc},\bnu}.
\]
The \ac{PIA} is a special case of the general occupancy-based weight-selection framework in Algorithm~\ref{alg:general_weight_framework}. It follows the same
iterative structure: starting from the current state-relevance weights, it generates a policy, computes that policy's normalized discounted
occupancy measure, and uses the resulting occupancy to update the weights. For the \ac{PIA}, the policy generator and weight-update rule
in Algorithm~\ref{alg:general_weight_framework} are
\[
\mathcal G_{\mathrm{PIA}}(\bc)=u_{\bc},
\qquad
\mathcal U_{\eta}(\bc,\bmu)
=
(1-\eta)\bc+\eta\bmu,
\]
where $u_{\bc}$ is the deterministic greedy policy obtained from the
primal \ac{ALP} solution under weights $\bc$.

Accordingly, at iteration $m$, the \ac{PIA} solves the primal \ac{ALP} using $\bc^{(m)}$, constructs the greedy policy $u_{\bc^{(m)}}$, computes its normalized discounted occupancy, and updates 
\begin{equation}
\label{eq:pia_update}
\bc^{(m+1)}
=
(1-\eta)\bc^{(m)}
+
\eta\bmu_{u_{\bc^{(m)}},\bnu},
\qquad
\eta\in(0,1].
\end{equation}
The original iteration of \cite{le2004choose} corresponds to $\eta=1$, whereas \cite{de2008choosing} gives a damped variant with $\eta<1$. Following Algorithm~\ref{alg:general_weight_framework}, the procedure stops when 
$
\left\|
\bmu_{u_{\bc^{(m)}},\bnu}-\bc^{(m)}
\right\|_{\infty}
\leq\varepsilon
$
or when the iteration limit is reached.

Thus, the \ac{PIA} and \ac{DIA} share the same general occupancy-based updating procedure. Their principal difference lies in policy generation: the \ac{PIA} constructs a deterministic greedy policy from the primal \ac{ALP} solution, whereas the \ac{DIA} constructs a smooth randomized policy directly from a selected dual solution. Unlike the smooth dual-based map analyzed in Section~\ref{sec:convergence}, the greedy-policy map used by the \ac{PIA} can be discontinuous.
Consequently, its weight iterates may not converge.

\section{Numerical details}
\label{sec:numerical_details}

All linear programming formulations were implemented in Python using Gurobi version 12.0.3 \citep{gurobi}. Experiments were run on Compute Canada's Narval cluster; each node includes a 64-core AMD EPYC 7532 (Zen 2) at 2.40 GHz, 256M cache L3, and 250 GB of RAM. 

Reported runtimes correspond to policy-generation time. All costs are reported as sample means with 95\% confidence intervals.

\section{Queueing Experiments Settings}
\label{app:queue}

\subsection{Single-Dimensional Queueing System}
\label{app:single_queue}

We consider the single-dimensional controlled queueing example of \citet{de2003linear}. In each period, a job arrives with probability $p=0.2$. The controller selects a service-rate action
\[
a\in\mathcal A
=
\{0,0.2,0.4,0.6,0.8\}.
\]
The one-period cost incurred in state $x$ under action $a$ is
\[
g(x,a)=x+60a^3.
\]
The queue capacity is $S=49{,}999$, so that
\[
\mathcal S=\{0,1,\ldots,49{,}999\},
\qquad
|\mathcal S|=50{,}000.
\]
The discount factor is $\alpha=0.98$. The approximation contains the following four basis functions
\[
\phi_1(x)=1,
\qquad
\phi_2(x)=x,
\qquad
\phi_3(x)=x^2,
\qquad
\phi_4(x)=x^3.
\]

\subsection{Eight-Dimensional Queueing Network}
\label{app:8d_queue}
Figure~\ref{fig:8_d_queue_demo_app} illustrates the eight-dimensional queueing network from the examples studied by \citet{de2003linear} and \citet{de2008choosing}. The state is the queue-length vector
\[
\bx=(x_1,\ldots,x_8)\in\mathbb Z_+^8.
\]
Jobs arrive sequentially at queues~1 and~4 with probabilities $\lambda_1$ and $\lambda_4$, respectively. When a job from queue $i$ is processed, service is completed with probability $\mu_i$, for $i=1,\ldots,8$. Following a service completion, the job is either routed to a downstream queue or leaves the system.

Let $\ba=(a_1,\ldots,a_8)\in\{0,1\}^8$, where $a_i=1$ indicates that the corresponding server is assigned to queue $i$. The feasible server allocations satisfy
\[
a_1+a_3+a_8=1,
\qquad
a_2+a_6=1,
\qquad
a_4+a_5+a_7=1.
\]
Thus, each server is assigned to exactly one queue in its server-sharing group at every decision epoch. 
The one-period immediate cost is the total queue length,
\[
g(\bx)=\sum_{i=1}^8x_i=\|\bx\|_1,
\]
and the discount factor is $\alpha=0.995$. The basis functions contain all monomials in $\bx$ of total degree at most two: one constant term, eight linear terms, and 36 quadratic terms, giving a total of $K=45$ basis functions.

Because the state space $\mathbb Z_+^8$ is countably infinite, the \ac{ALP} is constructed using Bellman constraints generated from $5{,}000$ distinct sampled states. Following \citet{de2003linear}, these states are sampled from the initial distribution,
\[
\nu_{\xi}(\bx)
=
(1-\xi)^8\xi^{\|\bx\|_1},
\qquad
\bx\in\mathbb Z_+^8,
\]
with $\xi=0.9$. 

\begin{figure}[htbp]
  \centering
  \includegraphics[width=0.8\linewidth]{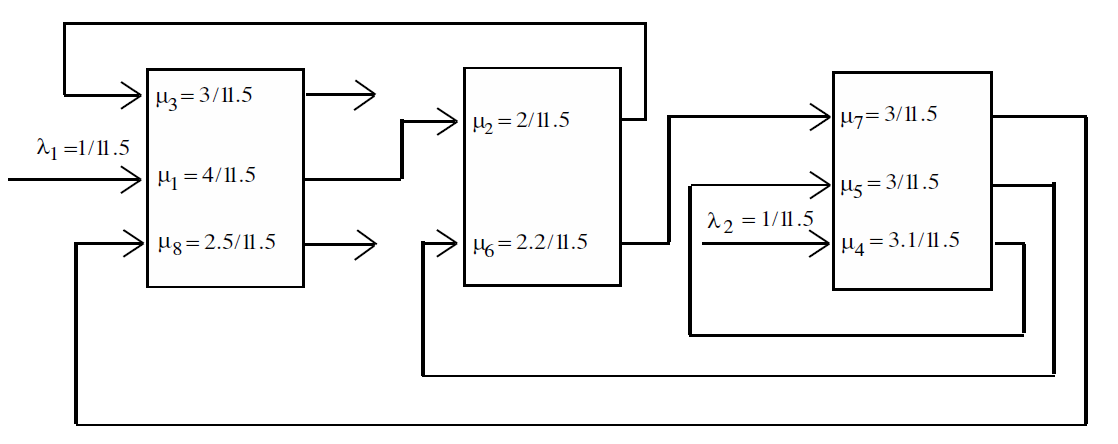}
  \caption{Eight-dimensional queueing network topology. Arrows show the routing of jobs between queues after service. Server-sharing groups are indicated by rectangles.}
  \label{fig:8_d_queue_demo_app}
\end{figure}

\section{Occupancy Rate and Convergence Insights}
\label{app:convergence_insights}

Figure~\ref{fig:discounted_occup} reports the normalized discounted occupancy measures in the single-dimensional queue. Neither fixed-weight policy matches the optimal occupancy across all initial distributions. In contrast, the \ac{PIA}- and \ac{DIA}\-induced policies match the optimal occupancy under \textit{Geometric} and closely approximate it under \textit{Empty-System}. Under the \textit{Reverse-Geometric} distribution, the system starts in a highly congested state and rarely visits queue lengths below 50; all policies assign negligible occupancy to those states. 
\begin{figure}[htbp]
  \centering
  \begin{subfigure}[b]{0.3\textwidth}
    \includegraphics[width=\linewidth]{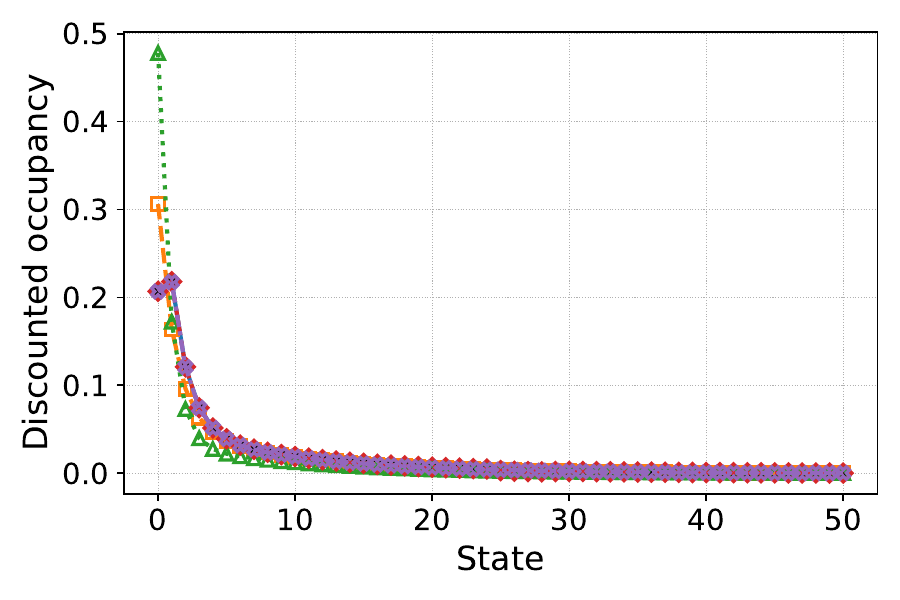}
    \caption{Geometric}
    \label{fig:discounted_occup_geo}
  \end{subfigure}\hfill
  \begin{subfigure}[b]{0.3\textwidth}
    \includegraphics[width=\linewidth]{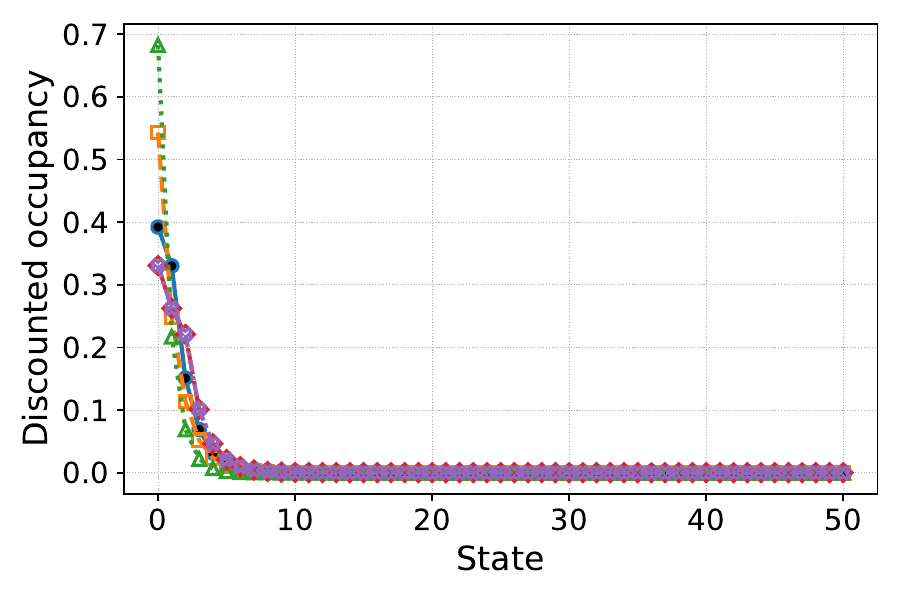}
    \caption{Empty}
    \label{fig:discounted_occup_empty}
  \end{subfigure}\hfill
  \begin{subfigure}[b]{0.3\textwidth}
    \includegraphics[width=\linewidth]{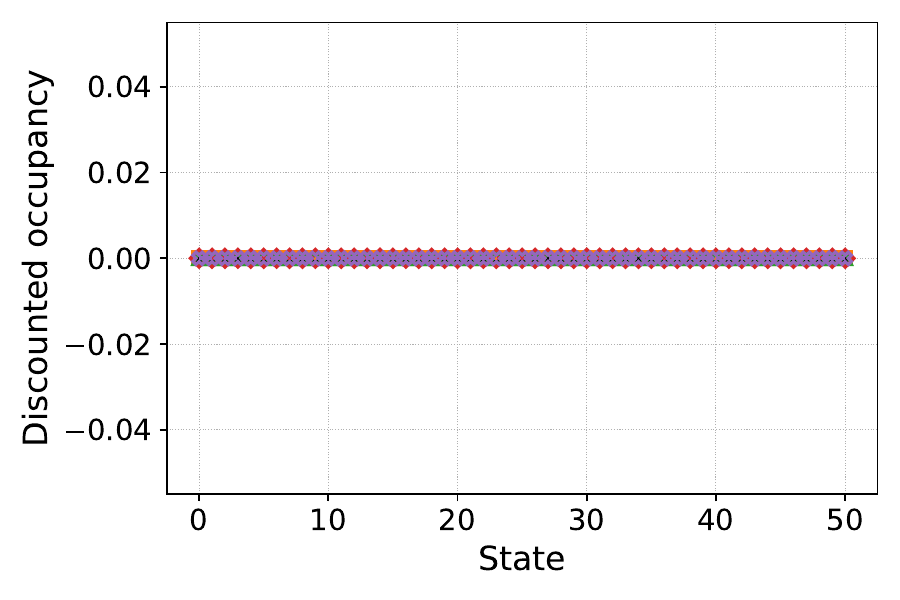}
    \caption{Reverse Geometric}
    \label{fig:discounted_occup_rev}
  \end{subfigure}

  \begin{subfigure}[b]{0.7\textwidth}
    \centering
    \includegraphics[width=\linewidth]{Fig_legend_bw.pdf}
  \end{subfigure}
  \caption{Normalized discounted occupancies under three initial state distributions.}
  \label{fig:discounted_occup}
\end{figure}

Figure~\ref{fig:8d_queue_convergence} reports the successive-iterate difference $\|\bc^{(m)\top}\Phi-\bc^{(m-1)\top}\Phi\|_1$. As the figure shows, \ac{DIA} moves steadily toward a fixed point, while \ac{PIA} oscillates. The oscillation reflects the non-smooth dependence of the greedy policy on the \ac{ALP} weight vector. A small change in weights can simultaneously change the greedy action in multiple states, causing a discontinuous shift in the occupancy measure used for the next iteration. Therefore, as the state space grows, \ac{DIA} (Random) captures most of the benefits of adaptive weighting at a much lower computational cost than \ac{PIA}, while \ac{DIA} ($\xi = 0.9$) further lowers mean cost but essentially forfeits that computational advantage.
\begin{figure}[htbp]
  \centering
  \includegraphics[width=0.5\linewidth]{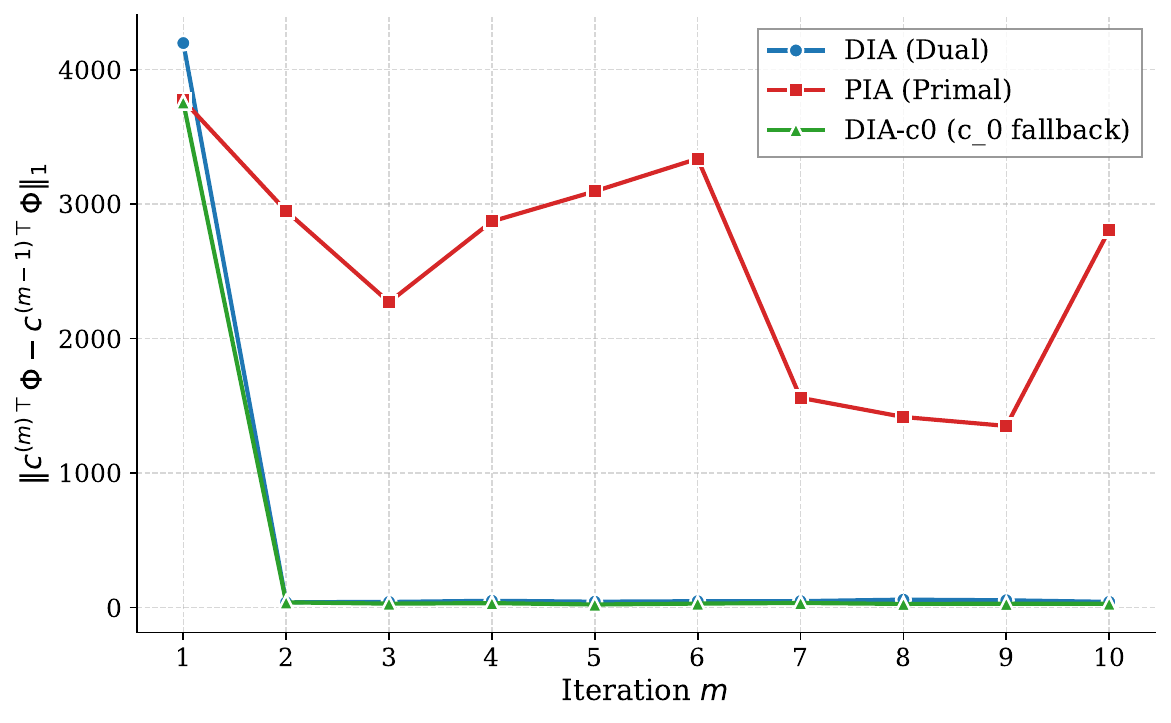}
  \caption{Convergence of \ac{PIA} and \ac{DIA} for the eight-dimensional queue: $\|\boldsymbol{c}^{(m)\top}\Phi - \boldsymbol{c}^{(m-1)\top}\Phi\|_1$ across iterations.}
  \label{fig:8d_queue_convergence}
\end{figure}

\section{Multi-Priority Patient Scheduling \ac{MDP}}
\label{app:scheduling}
The multi-priority advance-scheduling problem described in \citet{patrick2008dynamic} is modeled as a discrete time, infinite-horizon discounted-cost \ac{MDP}. The system state is
\begin{equation*}
\boldsymbol{s} \;=\; (\boldsymbol{x},\boldsymbol{y})
             \;=\; (x_1,\dots,x_N;\; y_1,\dots,y_I),
\end{equation*}
where $x_n$ denotes the number of patients already booked on day $n$ of an $N$-day booking horizon, and $y_i$ is the number of waiting patients in priority class $i$ for $i=1,\dots,I$. The state space is
\begin{equation*}
\mathcal{S}
\;=\;
\left\{(\boldsymbol{x},\boldsymbol{y}) \in \mathbb{Z}_+^{N}\times \mathbb{Z}_+^{I}
\,:\,
0 \le x_n \le C_1 \ \text{for } n=1,\dots,N,\quad
0 \le y_i \le Q_i \ \text{for } i=1,\dots,I
\right\},
\end{equation*}
where $C_1$ is the daily base capacity and $Q_i$ is an upper bound on the number of priority-$i$ arrivals in a day. In the experiments, $Q_i$ is chosen
large enough that demand truncation has negligible practical impact.

\textbf{Action space}. In each decision epoch, the scheduler books $a_{in}$ priority-$i$ patients into appointment slots on day $n$ and diverts $z_i$ patients in class $i$ to diversion capacity. The feasible action set is:
\[
\mathcal{A}_s
=
\Bigl\{(a,z)\in \mathbb{Z}_+^{I \times N}\times \mathbb{Z}_+^{I}:\;
x_n+\textstyle\sum_{i=1}^I a_{in}\le C_1\ \forall n;\
\sum_{i=1}^I z_i\le C_2;\
\sum_{n=1}^N a_{in}+z_i\le y_i\ \forall i
\Bigr\}.
\]
The first set of constraints enforces the regular daily capacity: because $x_n$ patients are already booked on day $n$, at most $C_1-x_n$ additional patients can be scheduled on that day. The second constraint limits the total number of patients assigned to diversion capacity to $C_2$. The third set of constraints ensures that, for each priority class, the total number of patients scheduled in regular appointment slots or diverted does not exceed the number $y_i$ currently waiting. Patients not assigned through either option remain in the waiting list. Finally, all action variables are restricted to nonnegative integers.

\textbf{Transition dynamics}.
After taking action $(a,z)$, new arrivals $\boldsymbol{y}'=(y'_1,\ldots,y'_I)$
are drawn independently, with $y'_i\sim p_i(\cdot)$ and
$p(\boldsymbol{y}')=\prod_{i=1}^{I} p_i(y'_i)$.
The booking horizon rolls forward and the state updates:
\begin{equation*}
\begin{gathered}
(x_1,\ldots,x_N;\,y_1,\ldots,y_I)
\to
\Bigl(x_2+\textstyle\sum_i a_{i2},\,\ldots,\,
x_N+\textstyle\sum_i a_{iN},\,0;\\
y'_1+y_1-\textstyle\sum_n a_{1n}-z_1,\,\ldots,\,
y'_I+y_I-\textstyle\sum_n a_{In}-z_I\Bigr).
\end{gathered}
\end{equation*}

\textbf{One-period cost}.
The cost accounts for (i) late-booking penalties, (ii) diversion costs, and (iii) penalties for waiting list:
\[
c(\boldsymbol{a},\boldsymbol{z})
=\sum_{i,n} b_i(n)\,a_{in}
\;+\;\sum_{i} d_i\, z_i
\;+\;\sum_{i} f_i\,\Bigl(y_i-\sum_{n} a_{in}-z_i\Bigr),
\]
where the booking penalty reflects the class-$i$ wait-time target $T(i)$:
\[
b_i(n)=
\begin{cases}
0, & n<T(i),\\[3pt]
\displaystyle\sum_{k=1}^{n-T(i)} f_i, & n>T(i).
\end{cases}
\]

\textbf{Linear \ac{VFA}}.
\cite{patrick2008dynamic} adopt the approximation
\[
\hat{J}(\boldsymbol{s})
= W_0 \;+\; \sum_{n=1}^{N} V_n\,x_n \;+\; \sum_{i=1}^{I} W_i\,y_i,
\qquad V_n \ge 0,\; W_i \ge 0,
\]
which reduces the \ac{ALP} to $N+I+1$ primal variables. Because the constraint set remains prohibitively large, they work with the dual formulation and solve it through column generation.

\scriptsize
\begin{align} 
\min_{\boldsymbol{\lambda}\,\geq\,0}\quad
& \sum_{(\boldsymbol{s},\boldsymbol{a},\boldsymbol{z})}
   \lambda(\boldsymbol{s},\boldsymbol{a},\boldsymbol{z})\,
   c(\boldsymbol{a},\boldsymbol{z})\notag \\[4pt] 
\text{s.t.}\quad
& (1-\alpha)
  \sum_{(\boldsymbol{s},\boldsymbol{a},\boldsymbol{z})}
  \lambda(\boldsymbol{s},\boldsymbol{a},\boldsymbol{z}) \;=\; 1,
  \label{eq:app_dual_norm} \\[4pt]
& \sum_{(\boldsymbol{s},\boldsymbol{a},\boldsymbol{z})}
  \lambda(\boldsymbol{s},\boldsymbol{a},\boldsymbol{z})
  \Bigl( x_n - \alpha x_{n+1} - \alpha \sum_{i} a_{i,n+1} \Bigr)
  \;\ge\; \mathbb{E}_{\eta}[X_n],
  \quad \forall n=1,\ldots,N,
  \label{eq:app_dual_x} \\[4pt]
& \sum_{(\boldsymbol{s},\boldsymbol{a},\boldsymbol{z})}
  \lambda(\boldsymbol{s},\boldsymbol{a},\boldsymbol{z})
  \Bigl( (1-\alpha)\,y_i
  + \alpha \bigl( \sum_{n} a_{in}+z_i-\mathbb{E}[Y_i] \bigr) \Bigr)
  \notag\\
&\hfill \ge \mathbb{E}_{\eta}[Y_i],
  \quad \forall i=1,\ldots,I,
  \label{eq:app_dual_y}
\end{align}
\small
where $\mathbb{E}_{\eta}[\cdot]$ is the expectation under a marginal distribution $\eta$ on the state space. The right-hand sides of \eqref{eq:app_dual_x}--\eqref{eq:app_dual_y} are the projected state-relevance weights $\bc^{\top}\bPhi$.

\textbf{Projection-Based \ac{DIA}}.
Because the state-relevance weight vector $\bc$ is high-dimensional, \ac{DIA} is recast in feature space: at each iteration, instead of updating
$\bc$ directly, we update its projection onto $\bPhi$, i.e.,
\[
  \bc^{\top}\bPhi \;=\; \left(1,\mathbb{E}_{\eta}\!\left[\bPhi(X,Y)\right]\right)
  \;=\;
  \bigl(1, \mathbb{E}_{\eta}[X_1],\,\ldots,\,\mathbb{E}_{\eta}[X_N],\,
        \mathbb{E}_{\eta}[Y_1],\,\ldots,\,\mathbb{E}_{\eta}[Y_I]\bigr)^{\top},
\]
and terminate when this projected vector has converged \citep{de2008choosing}. This retains all information required by \eqref{eq:app_dual_x}--\eqref{eq:app_dual_y} while avoiding explicit manipulation of a prohibitively large $\bc$.

\textbf{Numerical settings: small-case scenario.} Following \citet{patrick2008dynamic}, the small-case scenario has a daily regular capacity of $C_1=10$ appointment slots and a daily diversion capacity of $C_2=4$. The system contains $I=3$ priority classes with wait-time targets of $(7,14,21)$ days and uses a $N=30$-day booking horizon. Daily arrivals are independent Poisson random variables with means $(5,3,2)$ and are truncated at three times their respective means, yielding maximum daily arrivals of $(15,9,6)$. Thus, the mean total daily demand is 10, equal to the regular daily capacity. The overtime cost is $d=100$ per scan, the late-booking penalty parameters are $(20,10,5)$, and the discount factor is $\alpha=0.99$. The linear approximation contains $K=N+I+1=34$ basis functions. Each policy is evaluated using $11{,}000$ Monte Carlo episodes of horizon $H=688$, chosen so that $\alpha^H=0.99^{688}<10^{-3}$. For the initial-state distribution $\bnu$, each $x_n$ follows a discrete uniform distribution on $\{0,\ldots,C_1\}$, while each $y_i$ follows the corresponding Poisson distribution truncated to $\{0,\ldots,Q_i\}$.

\textbf{Numerical settings: large-case scenario.} The large-case scenario increases the regular daily capacity to $C_1=60$ scans while retaining a daily diversion capacity of $C_2=4$. Daily arrivals for the three priority classes are independent Poisson random variables with means $(10,20,30)$ and are truncated at three
times their respective means. The mean total daily demand is therefore 60, equal to the regular daily capacity, with a larger proportion of demand arising from the lower-priority classes. All remaining parameters are unchanged from the small-case scenario: the wait-time targets are $(7,14,21)$ days, the booking horizon is $N=30$ days, the overtime cost is $d=100$ per scan, the late-booking penalties are $(20,10,5)$, and the discount factor is $\alpha=0.99$. The linear approximation  again contains $K=34$ basis functions. The initial distribution follows similarly: each $x_n$ follows a discrete uniform distribution on $\{0,\ldots,C_1\}$, while each $y_i$ follows the corresponding Poisson distribution truncated to $\{0,\ldots,Q_i\}$.

\textbf{Numerical settings: tiny-case scenario.} To evaluate higher-order basis functions, we restrict the state space to a tractable size so that the optimization problem can be solved directly rather than through column generation. With higher-order polynomial basis functions, the column-generation pricing problem generally becomes nonlinear and may require prohibitively long solution times. To this end, we consider a tiny instance with a daily regular capacity of $C_1=2$ appointment slots and a daily diversion capacity of $C_2=1$. The system contains $I=2$ priority classes with wait-time targets of $(1,2)$ days and uses an $N=2$-day booking horizon. Daily arrivals are independent Poisson random variables with means $(2,1)$ and are truncated at three times their respective means, yielding maximum daily arrivals of $(6,3)$. The waiting-cost parameters are $(100,1)$, the diversion costs are $(200,5)$, and the discount factor is $\alpha=0.99$. The resulting state space contains $3^2\times7\times4=252$ states. For the initial-state distribution $\bnu$, each $x_n$ follows a discrete uniform distribution on $\{0,\ldots,C_1\}$, while each $y_i$ follows a discrete uniform distribution on $\{0,\ldots,Q_i\}$. We compare complete polynomial approximations of degrees one and four, containing $K=5$ and $K=70$ basis functions, respectively.

\end{APPENDICES}

\end{document}